\documentclass[a4paper,10pt]{article}
\usepackage{amstext,amsmath,amssymb,mathrsfs}
\usepackage{amsthm}
\usepackage{bm}
\usepackage{graphicx}
\usepackage{here}
\usepackage{setspace}
\usepackage{cases}

\renewcommand\hat{\widehat}

\def\a{\alpha}

\def\b{\beta}

\def\e{\epsilon}
\def\E{\mathbb{E}}
\def\g{\gamma}

\def\i{\infty}

\def\l{\lambda}

\def\P{\mathbb{P}}

\def\1{\bf{1}}
\def\R{\mathbb{R}}

\def\t{\tau}

\def\Z{\mathbb{Z}}

\def\1{{\bf 1}}

\newtheorem{theorem}{Theorem}[section]
\newtheorem{lemma}[theorem]{Lemma}
\newtheorem{proposition}[theorem]{Proposition}

\newtheorem{definition}[theorem]{Definition}

\theoremstyle{definition}
\newtheorem{remark}[theorem]{Remark}

\makeatletter
 
 \@addtoreset{equation}{section}
\makeatother

\title{The KPZ fixed point for discrete time TASEPs}
\author{Yuta Arai \thanks{Graduate School of Science and Engineering, Chiba University, Chiba-shi 263-8522, Japan. Email: yutaarai@chiba-u.jp}}
\date{}

\begin{document}

\maketitle

\begin{abstract}
We consider two versions of discrete time totally asymmetric simple 
exclusion processes (TASEPs) with geometric and Bernoulli hopping 
probabilities.  For the process mixed with these and continuous 
time dynamics, we obtain a single Fredholm determinant representation for the joint distribution function of particle positions with arbitrary initial data. 
This formula is a generalization of the recent result by Mateski, Quastel and Remenik and allows us to take the KPZ scaling limit.
For both the discrete time geometric and Bernoulli TASEPs, we show that 
the distribution functions converge to the one describing the KPZ fixed point.
\end{abstract}

\section{Introduction}
The totally asymmetric simple exclusion process (TASEP) is a prototypical interacting stochastic particle system and can be interpreted as a stochastic growth model of an interface, which turns out to belong to the Kardar-Parisi-Zhang (KPZ) universality class introduced in \cite{Kardar}. In addition, the TASEP is one
of the most basic models in the integrable probability~\cite{BoGo}. Remarkable algebraic structures allow us to obtain exact explicit forms of distribution functions
for some quantities.

On a macroscopic level, the particle density evolves deterministically according to the Burgers equation 
\cite{Rezakhanlou, Rost}. 
Therefore, a natural question arises: what kind of characteristic does the fluctuation around the deterministic growth have?  It has been know that it exhibits 
universal properties characterizing the KPZ class.
There are many important results in the literature of the integrable probability.
First, for the step initial condition, the one-point limiting distribution 
for the particle current in the TASEP has been obtained by Johansson~\cite{KJohansson} 
by converting the problem to the last passage percolation 
and then using the RSK correspondence.  It turned out that the limiting distribution 
is the GUE Tracy-Widom distribution. In~\cite{NaSaTASEP, Rakos}, this result has also been obtained by using an explicit determinantal form of the transition probability in the TASEP~\cite{Schutz}.
For the last passage problems with symmetries, similar results have been found by Baik-Rains~\cite{BRduke}. The results include the one-point limiting distribution of the particle current
for the alternating initial condition in the language of the TASEP
or equivalently, the height distribution for the flat initial condition in the language
of the growth process called the polynuclear growth (PNG) 
model~\cite{Prahofer}. In this case, the limiting distribution
turned out to be the GOE Tracy-Widom distribution.

These results on the one-point fluctuations have been generalized to the
case of the multi-point fluctuations. For the case corresponding to the 
step initial condition, a Fredholm determinant formula for 
the limiting multi-point distributions has been first 
obtained in the PNG model with space-time continuous 
setting~\cite{PrSpAiry2} by using the technique 
related to the RSK correspondence. The same result has been obtained for 
the space-time discretized PNG model~\cite{Johansson}. 
The limiting process characterized
by the multi-point distribution is called the Airy$_2$ process.
On the other hand, for the other initial conditions, the first important result
has been given in~\cite{TSasamoto}. Sasamoto  has developed 
the technique for obtaining the mult-point function in terms of 
the transition probability in the TASEP~\cite{Schutz} not only for the step initial
condition but also for the alternating one and has obtained 
a Fredholm determinant formula for the limiting functions in the alternating case. 
The process characterized by the multi-point distribution is now called
the Airy$_1$ process.
This approach in~\cite{TSasamoto} has been further studied 
and been applied to the TASEP and the PNG model with different 
settings~\cite{BoFe,Borodin,Sasamoto,Ferrari}.

We have been interested in the entire structure of the 
universal limiting process for more general initial data.
Our understanding of this problem has been advanced 
by the recent result by Matetski, Quastel, and Remenik~\cite{Matetski}.
Their result is based on the approach in~\cite{Sasamoto, TSasamoto}:
A Fredholm determinant formula for the distribution functions with
an arbitrary initial data has already been obtained in~\cite{Sasamoto} based on the approach developed in~\cite{TSasamoto}. The correlation kernel for the Fredholm 
determinant can be expressed in terms of the biorthogonal functions, say
$\Phi_k(x)$ and $\Psi_k(x)$.
The problem is that one of them, say $\Phi_k(x)$ does not have an explicit representation while $\Psi_k(x)$ does. Thus it had not been clear how to take the KPZ scaling limit of this kernel. \cite{Matetski} has overcome this situation. They represent the function in terms of a stopping time of the random walk with 
jumps obeying the geometric distributions.
This expression allows us to take the KPZ scaling limit since 
by Donsker\rq{}s invariance principle,
we easily find the stopping time converges to the one for the Brownian motion in the limit. Based on this technique, the limiting multi-point 
distribution functions for the particle positions in the arbitrary initial condition
has been obtained. The process with this multi-point distribution is called 
the KPZ fixed point. Recently various interesting progresses on this problem have 
been made for example in~\cite{Mihai, Leandro, DRemenik}.

In this paper, we show that the technique in~\cite{Matetski} can be applicable
to different versions of the TASEPs  besides the usual continuous time one.
In particular, we focus on two versions of the discrete time TASEPs:
the case where the random jump at each time step follows the (truncated)
geometric distribution and the parallel update is applied and also the case
where the random jump follows the Bernoulli distribution and
the (backward) sequential update is applied. Furthermore, in both cases, we consider
the situation where the hopping probabilities are time-dependent.
For the step initial condition, these dynamics have appear as a special case
of the higher spin vertex model and have been recently 
studied in~\cite{Alisa}. To the best of our knowledge, however,  
the analyses for the arbitrary initial condition has not been 
studied yet. We show Sch$\ddot{\rm u}$tz's type determinantal formulas 
for transition probabilities for both the geometric and Bernoulli TASEP
with time dependent hopping probabilities. Combining these with
Sch$\ddot{\rm u}$tz's formula for the continuous time TASEP, 
we get the determinantal transition probability for the system mixed
with the three types of dynamics. Using this, we obtain a Fredholm 
determinant formula for the multi-point distribution for the particle 
positions, in which we can take the KPZ scaling limit. This is a generalized
formula to the one~\cite{Matetski}: When we vanish the whole parameters
of the mixed dynamics except the part of the continuous time TASEP,
the determinantal formula is reduced to the result in~\cite{Matetski}.
Finally taking the KPZ scaling limit for both discrete time geometric 
and Bernoulli TASEP, we see that the multi-point distribution functions
converge to the one describing the KPZ fixed point. 

The paper is organized as follows. In~Sec.~\ref{s2}, we state the three versions
of the TASEPs, continuous time and two types of discrete time versions:
geometric and Bernoulli hopping. Their mixed version is also stated.
We also give our main result: the Fredholm determinant formula for the mixed TASEP (Theorem~\ref{Main}) and the KPZ scaling limit in two cases of
the geometric and Bernoulli TASEPs (Theorem~\ref{special}, and Propositions~\ref{scaling} and~\ref{gscaling}). In Sec.~\ref{s3}, after 
giving the determinantal formulas for the transition probabilities for the above
three types of TASEPs, we give the proof of Theorem~\ref{Main} using the framework developed in~\cite{Matetski}. In Sec.~\ref{asy}, we give  proofs
of Theorem~\ref{special}, and Propositions~\ref{scaling} and~\ref{gscaling}.
The crucial step is the saddle point analysis for the kernels.

\section{Models and results}
\label{s2}
In this section we define three versions of the TASEP and introduce 
our main results.

\subsection{Models}
In this paper we consider the TASEPs on $\Z$. Each particle jumps only to the right independently and stochastically 
if the target site is empty. If the site is occupied by the other particle, it cannot move, 
which represents the exclusion interaction.

In the TASEPs we mainly focus on the position of each particle.  Let $X_t(i)\in\Z$ be a position of the $i$th particle at time $t$. We set
$t\in\mathbb{Z}$ or $t\in\mathbb{R}$ according to the version. Since the dynamics of the TASEPs preserves 
the order of the particles, we can always assume
\\
\centerline{$\displaystyle\cdots<X_{t}(2)<X_{t}(1)<X_{t}(0)<X_{t}(-1)<X_{t}(-2)<\cdots$.}
\\
The particles at $\pm\infty$ are playing no role in the dynamics when adding $\pm\infty$ into the state space.

In this paper, we deal with the following three versions of the TASEP.
As written in Lemmas~\ref{MR1},~\ref{MR2}, and~\ref{MR3} in Sec.~\ref{s31}, they have a common feature that the
transition probability for each model is written as a single determinant form.
\subsubsection{Continuous time TASEP}
\label{cTASEP}

The continuous time TASEP on $\mathbb{Z}$ was introduced in~\cite{spitzer} in the literature of mathematics. In 
this case $t\in\R_{\ge 0}$ and each particle independently attempts to jump to the right neighboring site at 
rate $\gamma\in\R_{\ge 0}$ provided this site is empty. It is a continuous time Markov process with the generator $L$ defined  as follows: 
Let $\eta=\{\eta(x): x\in\mathbb{Z}\}\in\{0,1\}^{\mathbb{Z}}$ be a particle configuration.
For $x\in\mathbb{Z}$, $\eta(x)=1$ means the site $x$ is occupied by a particle while $\eta(x)=0$ means it is empty. 
The generator $L$ acting on cylinder functions $f:\{0,1\}^{\mathbb{Z}}\rightarrow\mathbb{R}$ is defined by
$$\displaystyle(Lf)(\eta)=\gamma\sum_{x\in\mathbb{Z}}\eta(x)(1-\eta(x+1))(f(\eta^{x,x+1})-f(\eta))$$
where
\begin{equation*}
\eta(x)=
\begin{cases}
1,&\text{if the site is occupied by a particle,}\\
0,&\text{if the site $x$ is empty,}
\end{cases}
\end{equation*}
and
$\eta^{x,x+1}$ denotes the configuration $\eta$ with the occupations at site $x$ and $x+1$ have been interchanged, that is,
\begin{equation*}
\eta^{x,x+1}(y)=
\begin{cases}
\eta(x+1)& \text{for~} y=x,\\
\eta(x)& \text{for~} y=x+1,\\
\eta(y)& \text{otherwise}.
\end{cases}
\end{equation*}
\subsubsection{Discrete time Bernoulli TASEP with sequential update}
\label{bTASEP}
We define the discrete time Bernoulli TASEP with sequential update on $\mathbb{Z}$.
This version was studied previously in \cite{aniso} as a marginal of dynamics on Gelfand-Tsetlin patterns which preserve the class of Schur processes
and more recently in~\cite{qtasep,Alisa} in the studies of the integrable probability. 

Let us assume the particle configurations at time $t\in\Z_{\ge0}$ as $X_t(j)=a_j,j\in\Z$.
The particle positions at time $t+1$ are determined by the following update rule:
We update the position of the $i$th particle $X_{t+1}(i)$ in increasing order.
Suppose that we already updated the $i-1$th particle and its position is $b_{i-1}$ 
i.e. $X_{t+1}(i-1)=b_{i-1}$.
Then the update rule is given as follows: 
\begin{itemize}
\item When $X_{t+1}(i-1)-X_t(i)=b_{i-1}-a_i>1$,
\begin{equation*}
\mathbb{P}(X_{t+1}(i)=a|X_{t}(i)=a_i,X_{t+1}(i-1)=b_{i-1})=
\begin{cases}
1-p_{t+1}&\text{for $a=a_i$,}\\
p_{t+1}&\text{for $a=a_i+1$,}\\
0&\text{otherwise.}
\end{cases}
\end{equation*}
\item When $X_{t+1}(i-1)-X_t(i)=b_{i-1}-a_i=1$,
\begin{equation*}
\mathbb{P}(X_{t+1}(i)=a|X_{t}(i)=a_i,X_{t+1}(i-1)=b_{i-1})=
\begin{cases}
1&\text{for $a=a_i$,}\\
0&\text{otherwise.}
\end{cases}
\end{equation*}
\end{itemize}
This dynamics mean that starting from right to left, for the time step $t\rightarrow t+1$,
the $i$th particle jumps to the right neighboring site with probability $p_{t+1}\in(0,1)$ provided this site is empty.
Since the update is sequential from right to left, during a time step, a block of consecutive particles can jump.
For later use, we define $\beta_t,~t=0,1,2,\dots$ by
\begin{align}
\label{pbeta}
p_t=\frac{\beta_t}{1+\beta_t},~\left(\beta_t=\frac{p_t}{1-p_t}\right).
\end{align}
\begin{remark}
In the case of discrete time Bernoulli TASEP with {\it parallel} update,
some integrable structures have also been studied for example in~\cite{Ferrari,IS2007, PP2007}. 
In~\cite{PP2007}, an explicit form of the transition probability was obtained by using the Bethe ansatz.
However, it is written as a ratio of two determinants not a single determinant.
To study the KPZ fixed point in this case is an interesting future problem.
\end{remark}
\subsubsection{Discrete time geometric TASEP with parallel update}
\label{aTASEP}
We define the discrete time geometric TASEP 
with parallel update on $\mathbb{Z}$.
This was studied previously in \cite{Jon} as a marginal of dynamics on Gelfand-Tsetlin patterns which preserve the class of Schur processes.
More recently it has been also investigated in~\cite{qtasep,Alisa}.

Let us assume that for $t\in\Z_{\ge 0}$ and $j\in\Z$, $X_t(j)=a_j$.
The update rule of the positions at time
$t+1$ are given as follows:
For each $1\leq i \leq N$,
\begin{multline}
\label{geo}
\mathbb{P}(X_{t+1}(i)=a_i+a|X_{t}(i)=a_i,~X_{t}(i-1)=a_{i-1})
\\
=
\begin{cases}
\alpha_{t+1}^a (1-\alpha_{t+1}) &\text{for $a=0,1,\dots, a_{i-1}-a_i-2$,}
\\
\alpha_{t+1}^a & \text{for $a=a_{i-1}-a_i-1$,}
\\
0 & \text{otherwise,}
\end{cases}
\end{multline}
where the update is independent for each $i$ and $t$.

Note that in this dynamics, the $j$th particle can jump
with multiple cites according to the truncated geometric distribution defined in~\eqref{geo} with parameter $\alpha_t$.
\begin{remark}
As shown in Lemma~\ref{MR3} below, we have a
determinantal formula for the transition probability in this model.
In the discrete time geometric TASEP with {\it sequential} update, 
it is not clear if it has any solvable structures via Bethe ansatz or
an explicit formula for the transition probability.
\end{remark}
\subsubsection{TASEP$_{\bm{\a,\b},\g}$: TASEP mixed with the continuous time TASEP and the discrete time TASEPs}
\label{abcTASEP}

In this paper, we consider the TASEP combined with the above three versions.
First we take three time parameters $t_1,~t_2\in \mathbb{Z}_{\geq 0}$
and $t_3\in\mathbb{R}_{\geq 0}$. Then particles evolve according to the discrete time geometric TASEP with parameter $\bm{\a}:=\{\a_1,\a_2,\dots,\a_{t_1}\}$(Sec.~\ref{aTASEP}) from time $0$ to $t_1$, the discrete time Bernoulli TASEP with 
parameter $\bm{\b}=\{\b_{t_1+1},\b_{t_1+2},\dots,\b_{t_1+t_2}\}$ (Sec.~\ref{bTASEP}) 
from time $t_1$ to $t_1+t_2$, and the continuous time TASEP with parameter $\gamma$ (Sec.~\ref{cTASEP}) from $t_1+t_2$ to $t_1+t_2+t_3$.
In this paper we denote
this mixed TASEP as ${\rm TASEP}_{\bm{\a,\b},\g}$.

This type of the mixed TASEP with $t_3=0$ has been 
introduced in~\cite{Alisa, OP17}.  A related process has been studied in~\cite{DPPP12}.
We decided the order of the three dynamics as above.  In fact the distribution of the particles' 
positions is invariant even if we freely exchange order of these dynamics 
since the semigroups of all the three dynamics are shown to be exchangeable 
thanks to the Yang-Baxter relations~\cite{BoPe, CoPe}.
\begin{remark}
The motivation of introducing the TASEP$_{\bm{\a,\b},\g}$ is that
we can treat the above three models in a unified way.
As stated in Proposition~\ref{abcabc} below, one can see that 
the transition probability of the TASEP$_{\bm{\a,\b},\g}$ is also written as a single determinant combining the determinantal formulas (Lemmas~\ref{MR1},
~\ref{MR2} and \ref{MR3}) for the above three TASEPs in Sec.~\ref{cTASEP}-\ref{aTASEP}. Starting from the determinantal formula,
one can generalize the approach to the continuous time TASEP
in~\cite{Matetski} to the TASEP$_{\bm{\a,\b},\g}$.
We will explain it in Secs.~\ref{s32}-\ref{s34}.
\end{remark}
\subsection{Results}
In this subsection, we give our main results.

\subsubsection{Joint distribution of the particle positions}
Here we give a single Fredholm determinant formula for joint distribution
of the particle position in TASEP$_{\bm{\a,\b},\g}$ defined in Sec.~\ref{abcTASEP}.
For the descriptions of the results below including the following one,  
we state some definitions.
\begin{definition}
For a real single-valued function, $\hat{f}:\mathbb{A}\to(-\infty, \infty]$
with (in general an uncountable) domain $\mathbb{A}$, the epigraph  ${\rm epi}(\hat{f})$
and the hypograph ${\rm hypo} (\hat{f})$
are defined as follows.
\begin{align*}
{\rm epi}(\hat{f})=\{(x,y) : y\geq \hat{f}(x)\},~{\rm hypo} (\hat{f})=\{(x,y) : y\leq \hat{f}(x)\}.
\end{align*}
\end{definition}

\begin{definition}
\label{defRWtau}
Let ${RW}_m,~m=0,1,2\dots$ be the position of a random walker with Geom$[\frac{1}{2}]$ jumps strictly to the left starting at some fixed site $c$, 
i.e.,
\begin{align*}
{RW}_m=c-\chi_1-\chi_2-\cdots -\chi_m,
\end{align*}
where $\chi_i,~i=1,2,\dots$ are the i.i.d. random variable with
$
\P(\chi_i=k)=1/2^{k+1},~k=0,1,2,\dots.
$

We also define the stopping time
\begin{equation}
\label{Tau}
\tau=\min\{m\geq 0 : RW_m>X_0(m+1)\},
\end{equation}
where $\tau$ is the hitting time of the strict epigraph of the curve $(X_0(k+1))_{k=0,\dots, n-1}$ by the random walk $RW_k$.
When the number of particles is N, $X_0(m)$ is constant
and defined only $m\le N$.
\end{definition}

At last we define the multiplication operators.
\begin{definition}
For a fixed vector $a\in\mathbb{R}^m$ and indices $n_1<\cdots<n_m$, we define $\chi_{a}$ and 
$\bar{\chi}_{a}$ by the multiplication operators acting on the space 
$\ell^2(\{n_1,\dots,n_m\}\times\mathbb{Z})$(or acting on the space $L^2(\{x_1,\dots,x_m\}\times\mathbb{R})$) 
with
\begin{align}
\label{chi}
\chi_{a}(n_j,x)=\1_{x>a_j},~~\bar{\chi}_{a}(n_j,x)=\1_{x\leq a_j}.
\end{align}
\end{definition}
We obtain the following result.
\begin{theorem}
\label{Main}
We consider the $\rm{TASEP}_{\bm{\a,\b},\g}$
introduced in Sec.~\ref{abcTASEP}. Let $t=t_1+t_2+t_3$ be the final time, and $X_t(j),~j\in\Z$ be the 
the position of the particle labeled $j$ at $t(=t_1+t_2+t_3)$.
Assume that the initial positions $X_0(j)\in\Z$ for $j=1,2,\dots$ 
are arbitrary constants satisfying $X_0(1)>X_0(2)>\cdots$
while $X_0(j)=\infty$ for $j\leq 0$.

For $n_j\in\Z_{\ge 1}~j=1,2,\dots,M$ with $1\leq n_1<n_2<\cdots<n_M$, 
and ${a}=(a_1,a_2, \dots, a_M)\in\Z^M$
we have
\begin{equation}
\label{pro}
\mathbb{P}(X_t(n_j)>a_j, j=1,\dots,M)=\det(I-\bar{\chi}_{a}K_t\bar{\chi}_{a})_{\ell^2(\{n_1,\dots,n_M\}\times\mathbb{Z})},
\end{equation}
where $\bar{\chi}_{\bm a}(n_j,x)$ is defined in~\eqref{chi} 
and the kernel $K_t$ is given by
\begin{align}
\label{Kt}
&K_t(n_i, x  ; n_j, y)=-Q^{n_j-n_i}(x,y)\1_{n_i<n_j}+(S_{-t, -n_i})^{*}\bar{S}^{{\rm epi} (X_0)}_{-t,n_j}(x,y),
\\
&Q^m(x,y)=\frac{1}{2^{x-y}}\binom{x-y-1}{m-1}\1_{x\geq y+m},
\\
\label{stn}
&S_{-t,-n}(z_1,z_2)=\frac{1}{2\pi i}\oint_{\Gamma_0}dw\frac{(1-w)^n}{2^{z_2-z_1}w^{n+1+z_2-z_1}} \mathfrak{F}_{{\a,\b},\g}(w, t),
\\
\label{bstn}
&\bar{S}_{-t,n}(z_1,z_2)=\frac{1}{2\pi i}\oint_{\Gamma_0}dw\frac{(1-w)^{z_2-z_1+n-1}}{2^{z_1-z_2}w^{n}} \bar{\mathfrak{F}}_{{\a,\b},\g}(w, t),
\\
\label{sepi}
&\displaystyle\bar{S}^{\rm epi(X_0)}_{-t,n}(z_1, z_2)=\mathbb{E}_{RW_0=z_1}\left[\bar{S}_{-t, n-\tau}(RW_{\tau}, z_2)\1_{\tau<n}\right],
\\
&\mathfrak{F}_{{\a,\b},\g}(w, t)=\prod_{j=1}^{t_1}\frac{1}{1-\frac{2\a_j}{2-\a_j} \left(w-\frac{1}{2}\right)}\cdot\prod_{j=t_1+1}^{t_1+t_2}\left\{1+\frac{2\b_j}{2+\b_j} \left(w-\frac{1}{2}\right)\right\}\cdot e^{\g t_3\left(w-\frac{1}{2}\right)},
\\
&\bar{\mathfrak{F}}_{{\a,\b},\g}(w, t)=\prod_{j=1}^{t_1}\left\{1+\frac{2\a_j}{2-\a_j} \left(w-\frac{1}{2}\right)\right\}\cdot\prod_{j=t_1+1}^{t_1+t_2}\frac{1}{1-\frac{2\b_j}{2+\b_j} \left(w-\frac{1}{2}\right)}\cdot e^{\g t_3\left(w-\frac{1}{2}\right)},
\end{align}
where
$\Gamma_0$ is a simple counterclockwise loop around $0$ 
not enclosing any other poles.
The superscript $\rm epi(X_0)$ in~\eqref{sepi} refers to the fact that $\tau$ is the hitting time of the strict epigraph of the curve $(X_0(k+1))_{k=0,\dots, n-1}$ by the random walk $RW_k$ (see Def.~\ref{defRWtau}).
\end{theorem}
\begin{remark}
In the case of continuous time TASEP, i.e. the special case $\a_i=\b_j=0$ with
$1\le i\le t_1$, $t_1+1\le j \le t_1+t_2$, this formula has been obtained in 
Theorem 2.6 in~\cite{Matetski}. Theorem \ref{Main} above is the
generalization of the result in~\cite{Matetski} to the $\rm{TASEP}_{\bm{\a,\b},\g}$,
which includes the two types of discrete time TASEPs as well 
as the continuous time one.
\end{remark}
\subsubsection{The Kardar-Parisi-Zhang (KPZ) scaling limit}
Here we state our result on the scaling limit of the joint distribution function in Theorem~\ref{Main}. 
Although we expect that the scaling limit 
can be taken for the general TASEP$_{\bm{\a,\b},\g}$, 
we analyze two simpler cases, the discrete time Bernoulli TASEP with sequential update and the discrete time geometric TASEP with parallel update in this paper
since the asymptotic analysis in the general case would be somewhat involved.

We focus on the following two cases:
\begin{itemize}
\item The discrete time Bernoulli TASEP (Sec.~\ref{bTASEP})

In the TASEP$_{\a,\b,\g}$ introduced in Sec.~\ref{abcTASEP}, the case is realized by 
the specialization
\begin{equation*}
\a_1=\a_2=\cdots=\a_{t_1}=\g=0,~\b_{t_1+1}=\b_{t_1+2}=\cdots=\b_{t_1+t_2}=\b=\frac{p}{1-p}.
\end{equation*}
\item The discrete time geometric TASEP (Sec.~\ref{aTASEP})

As above, it is realized by
\begin{align*}
\a_1=\a_2=\cdots=\a_{t_1}=\a,~\g=\b_{t_1+1}=\b_{t_1+2}=\cdots=\b_{t_1+t_2}=0.
\end{align*}
\end{itemize}

To see the universal behavior of the fluctuations, we focus on
the height function defined as follows.
\begin{definition}
For $z\in\mathbb{Z}$, the TASEP height function related to $X_t$ is given by
\begin{equation}
\label{height}
\displaystyle h_t(z)=-2(X^{-1}_t(z-1)-X^{-1}_0(-1))-z
\end{equation}
where 
\begin{equation}
X^{-1}_t(u)=\min\{k\in\mathbb{Z}:X_t(k)\leq u\}
\end{equation}
denote the label of the rightmost particle which sits to the left of, or at, $u$ at time $t$
and we fix $h_0(0)=0$.
\end{definition}

Note that it can be represented as
\begin{equation}
h_t(z+1)=h_t(z)+\hat{\eta}_t(z).
\end{equation}
where
\begin{equation*}
\hat{\eta}_t(z)
=\begin{cases}
1&\text{if there is a particle at $z$ at time $t$,}\\
-1&\text{if there is no particle at $z$ at time $t$.}
\end{cases}
\end{equation*}
We can extend the height function to a continuous function of $x\in\mathbb{R}$ by linearly interpolating between the integer points.

It is well known that the TASEP belongs to the Kardar-Parisi-Zhang (KPZ) universality class.
Thus we expect that the proper scaling of the height function is
\begin{align}
\frac{h_t(x)-At}{Ct^{\frac13}}, \text{~with~} x=Bt^{2/3}.
\label{kpzscaling1}
\end{align}
On average the height of the TASEP grows as $t^1$ with speed $A$, which is a constant.
On the other hand the fluctuation of the height around the average is of order $t^{1/3}$
contrary to the $t^{1/2}$ of the usual scaling in the central limit theorem. The scaling exponent
of the $x$-direction is $2/3$, the twice of the one in $h$-direction $1/3$, which suggest that
the path of the height function becomes the Brownian motion like. The exponents $(1/3, 2/3)$
are known to be universal and characterizing the KPZ universality class while the constants
$A,B,C$ are not universal and depend on the models. As shown in Sec.~\ref{asy}, we have 
\begin{itemize}
\item the discrete time Bernoulli TASEP case
\begin{align}
A= \frac{p-2}{2},B=2, C=1,
\label{kpzscaling2}
\end{align}
\item the discrete time geometric TASEP case
\begin{align}
A= \frac{\a-2}{2(1-\a)},B=2, C=1.
\label{kpzscaling3}
\end{align}
\end{itemize}
Based on the property of the height function, we define the scaled height, which is equivalent
to~\eqref{kpzscaling2} and~\eqref{kpzscaling3} but a slightly different form appearing 
as the ``1:2:3 scaling\rq\rq{} in~\cite{Quastel}.
\begin{definition}
For $\mathbf{t}\in\R_{\ge 0}$ and $\mathbf{x}\in\R$, we define the scaling height function as the following.
\begin{itemize}
\item The discrete time Bernoulli TASEP 
\begin{equation}
\label{Hei}
\displaystyle\hat{h}^{\varepsilon}(\mathbf{t},\mathbf{x})=\varepsilon^{\frac{1}{2}}\left[h_{t}(x)+\frac{2-p}{2}\varepsilon^{-\frac{3}{2}}\mathbf{t}\right],
\end{equation}
where $t$ and $x$ are scaled as
\begin{align}
\label{txscaling}
t=\frac{(2-p)^3}{4p(1-p)}\varepsilon^{-\frac{3}{2}}\mathbf{t},
~
x=2\varepsilon^{-1}\mathbf{x}.
\end{align}
\item The discrete time geometric TASEP
\begin{equation}
\label{gHei}
\displaystyle\hat{h}^{\varepsilon}(\mathbf{t},\mathbf{x})=\varepsilon^{\frac{1}{2}}\left[h_{t}(x)+\frac{2-\a}{2(1-\a)}\varepsilon^{-\frac{3}{2}}\mathbf{t}\right],
\end{equation}
where $t$ and $x$ are scaled as
\begin{align}
\label{txscaling}
t=\frac{(2-\a)^3}{4\a(1-\a)}\varepsilon^{-\frac{3}{2}}\mathbf{t},
~
x=2\varepsilon^{-1}\mathbf{x}.
\end{align}
\end{itemize}
\end{definition}
Our goal is to compute the $\varepsilon\rightarrow 0$ limit of the joint distribution function,
\begin{align}
\lim_{\varepsilon\xrightarrow{}0}\mathbb{P}_{\hat{h}^{\varepsilon}_0}(\hat{h}^{\varepsilon}(\mathbf{t}, \mathbf{x}_1)\leq\mathbf{a}_1,\dots, \hat{h}^{\varepsilon}(\mathbf{t}, \mathbf{x}_m)\leq\mathbf{a}_m)
\label{ldisfunc}
\end{align}
for $\mathbf{x}_1<\mathbf{x}_2< \dots <\mathbf{x}_m \in\mathbb{R}$ and $\mathbf{a}_1,\dots, \mathbf{a}_m \in\mathbb{R}$.
Here $\P_{\hat{h}^{\varepsilon}_0}(\cdot)$ represents the probability measure in which the initial height profile is
$\hat{h}^{\varepsilon} (0,x)$. We will show that the limit converges to the joint distribution function characterizing
the KPZ fixed point introduced in~\cite{Matetski}.

Here we introduce the KPZ fixed point. First we define UC and LC as follows.
\begin{definition}{\rm(UC} and {\rm LC}~{\rm \cite{Matetski})}.\\
We define {\rm UC} as the space of upper semicontinuous functions $\hat{h}:\mathbb{R}\to[-\infty, \infty)$ with $\hat{h}(x)\leq C_1+C_2|x|$ for some $C_1, C_2<\infty$ and $\hat{h}(x)>-\infty$ for some $x$
and {\rm LC} as ${\rm LC}=\{\hat{g} : -\hat{g}\in {\rm UC}\}$.
\end{definition}
Now we are ready to state the KPZ fixed point.
For more detailed information, see~\cite{Matetski}.
\begin{definition}[The KPZ fixed point~\cite{Matetski}]
The KPZ fixed point is the unique Markov process on {\rm UC}, $(\hat{h}(\mathbf{t},\cdot))_{\mathbf{t}>0}$
with transition probabilities given by 
\begin{equation}
\label{kpzfpdet}
\displaystyle\mathbb{P}_{\hat{h}_0}(\hat{h}(\mathbf{t},\mathbf{x}_1)\leq \mathbf{a}_1,\dots, \hat{h}(\mathbf{t},\mathbf{x}_m)\leq \mathbf{a}_m)=\det\left(\mathbf{I}-\chi_{\mathbf{a}}\mathbf{K}^{{\rm hypo}(\hat{h}_0)}_{\mathbf{t},{\rm ext}}\chi_{\mathbf{a}}\right)_{L^2(\{\mathbf{x}_1,\dots, \mathbf{x}_m\}\times\mathbb{R})}.
\end{equation}
Here in LHS, $\mathbf{x}_1<\mathbf{x}_2< \dots <\mathbf{x}_m \in\mathbb{R}$ and $\mathbf{a}_1,\dots, \mathbf{a}_m \in\mathbb{R}$, $\hat{h}_0\in{\rm UC}$ and $\mathbb{P}_{\hat{h}_0}$ means the measure on the process with initial data $\hat{h}_0$. In RHS, the kernel is given by
\begin{align}
\label{funct}
&~\mathbf{K}^{{\rm hypo}(\hat{h}_0)}_{\mathbf{t},{\rm ext}}(\mathbf{x}_i, v ;\mathbf{x}_j, u)
\notag
\\
&=-\frac{1}{\sqrt{4\pi(x_j-x_i)}}\exp\left(-\frac{(u-v)^2}{4(x_j-x_i)}\right)\1_{\mathbf{x}_i<\mathbf{x}_j}+\left(\mathbf{S}^{{\rm hypo}(\hat{h}^{-}_0)}_{\mathbf{t},-\mathbf{x}_i}\right)^{*}\mathbf{S}_{\mathbf{t}, \mathbf{x}_j}(v, u),
\\
&
\label{function}
~\displaystyle\mathbf{S}_{\mathbf{t}, \mathbf{x}}(v, u)
=\mathbf{t}^{-\frac{1}{3}}e^{\frac{2\mathbf{x}^3}{3\mathbf{t}^2}-\frac{(v-u)\mathbf{x}}{\mathbf{t}}} {\rm Ai}(-\mathbf{t}^{-\frac{1}{3}}(v-u)+\mathbf{t}^{-\frac{4}{3}}\mathbf{x}^2),
\\
\label{functione}
&
~\displaystyle \mathbf{S}^{{\rm hypo}(\hat{h})}_{\mathbf{t},\mathbf{x}}(v,u)=\mathbb{E}_{B(0)=v}[\mathbf{S}_{\mathbf{t}, \mathbf{x}-\bm{\tau}'}(B(\bm{\tau}'),u)\1_{\bm{\tau}'<\infty}],
\end{align}
where $(A)^*$ represents the adjoint of an integral operator $A$, and $B(x)$ is a Brownian motion 
with diffusion coefficient $2$ and $\bm{\tau}'$ is the hitting time of the hypograph of the function $\hat{h}$.
\end{definition}
\begin{remark}
(\ref{funct}) and (\ref{function}) can be written in terms of  
the differential operators 
$\mathbf{S}_{\mathbf{t}, \mathbf{x}}=\exp\{\mathbf{x}\partial^2+\mathbf{t}\partial^3/3\}$,
\begin{equation*}
\mathbf{K}^{{\rm hypo}(\hat{h}_0)}_{\mathbf{t},{\rm ext}}(\mathbf{x}_i, \cdot ;\mathbf{x}_j, \cdot)=-e^{(\mathbf{x}_j-\mathbf{x}_i)\partial^2}\1_{\mathbf{x}_i<\mathbf{x}_j}+\left(\mathbf{S}^{{\rm hypo}(\hat{h}^{-}_0)}_{\mathbf{t},-\mathbf{x}_i}\right)^{*}\mathbf{S}_{\mathbf{t}, \mathbf{x}_j}.
\end{equation*}
In addition using the integral representation for the Airy function
\begin{equation*}
\displaystyle{\rm Ai}(z)=\frac{1}{2\pi i}\int_{\langle}dw \ e^{\frac{1}{3}w^3-zw},
\end{equation*}
where $\langle$ is the positively oriented contour going the straight lines from $e^{-\frac{i\pi}{3}}\infty$ to $e^{\frac{i\pi}{3}}\infty$ through 0, we find that 
$\mathbf{S}_{\mathbf{t}, \mathbf{x}}(v, u)$~\eqref{function} can be expressed as
\begin{align}
\displaystyle\mathbf{S}_{\mathbf{t}, \mathbf{x}}(v, u)=\frac{1}{2\pi i}\int_{\langle}dw \ e^{\frac{\mathbf{t}}{3}w^3+xw^2-(v-u)w}.
\end{align}
\end{remark}

Now we assume that the limit
\begin{equation}
\label{Height}
\hat{h}_0=\lim_{\varepsilon\xrightarrow{}0}\hat{h}^{\varepsilon}(0, \cdot)
\end{equation}
exists. Note that by(\ref{height}) and (\ref{Height}), (\ref{initial}), this assumption is rewritten as 
\begin{equation}
\label{varx}
\displaystyle\varepsilon^{\frac{1}{2}}[(X^{\varepsilon}_0)^{-1}(\mathbf{x})+2\varepsilon^{-1}\mathbf{x}-2]\xrightarrow[\varepsilon\xrightarrow{}0]{}-\hat{h}_0(-\mathbf{x}),
\end{equation}
where $(X^{\varepsilon}_0)^{-1}(\mathbf{x}):=2X^{-1}_0(-2\varepsilon^{-1}\mathbf{x}-1)$ and the left hand side is interpreted as a linear interpolation to make it a continuous function of $x\in\mathbb{R}$ and we chose the frame of reference by
\begin{equation}
\label{initial}
X^{-1}_0(-1)=1,
\end{equation}
i.e. the particle labeled $1$ is initially the rightmost in $\mathbb{Z}_{<0}$.

Under this assumption, we have the following result for the limiting joint distribution function~\eqref{ldisfunc}.
\begin{theorem}(One-sided fixed point formula).
\label{special}
Let $\hat{h}_0\in{\rm UC}$ with $\hat{h}_0(\mathbf{x})=-\infty$ for $\mathbf{x}>0$.
Then given $\mathbf{x}_1<\mathbf{x}_2< \dots <\mathbf{x}_m \in\mathbb{R}$ and $\mathbf{a}_1,\dots, \mathbf{a}_m \in\mathbb{R}$, we have
\begin{equation}
\displaystyle\lim_{\varepsilon\xrightarrow{}0}\mathbb{P}_{\hat{h}^{\varepsilon}_0}(\hat{h}^{\varepsilon}(\mathbf{t}, \mathbf{x}_1)\leq\mathbf{a}_1,\dots, \hat{h}^{\varepsilon}(\mathbf{t}, \mathbf{x}_m)\leq\mathbf{a}_m)
=
\det\left(\mathbf{I}-\chi_{\mathbf{a}}\mathbf{K}^{{\rm hypo}(\hat{h}_0)}_{\mathbf{t}, {\rm ext}}\chi_{\mathbf{a}}\right)_{L^2(\{\mathbf{x}_1,\dots, \mathbf{x}_m\}\times\mathbb{R})},
\end{equation}
where RHS is equivalent to that of~\eqref{kpzfpdet}.
\end{theorem}
\begin{remark}
We only give pointwise convergence of the kernels.
In principle, one expects the
convergence could be upgraded to trace class (see~\cite{Matetski}, \cite{Remenik} and \cite{Simon}) which would give a full proof of
Theorem~\ref{special}.
\end{remark}
\begin{remark}
The One-sided fixed point formula for the continuous time TASEP has been given in Proposition 3.6 in \cite{Matetski}. Our theorem \ref{special} indicates that Bernoulli TASEP and geometric TASEP 
settle into the same class ``KPZ fixed point”. The KPZ fixed point is believed to be the universal
process for the KPZ class with arbitrary fixed initial data. Our result supports this universality.
\end{remark}

\begin{remark}
In fact we can remove the assumption $\hat{h}_0({\bf x})=-\infty$ for ${\bf x}>0$ in the above theorem
by using the similar argument in Theorem 3.8. in~\cite{Matetski}.
\end{remark}
To prove Theorem~\ref{special}, we use the following relationship between the particle positions $X_t(j)$
and the height function $h_t(z)$~\eqref{height}.
Let $s_1,\dots, s_k, m_1, \dots, m_k\in\mathbb{R}$ and $z_1,\dots, z_k, n_1,\dots, n_k\in\mathbb{Z}$.
We have
\begin{align}
\label{Xhrelation}
\P(h_t(z_1)\le s_1,\dots, h_t(z_k)\le s_k)=\P(X_t(n_1)\ge m_1,\dots, X_t(n_k)\ge m_k),
\end{align}
which follows from the definitions of  $h_t(x)$~\eqref{height}.
By this relation, we see
\begin{multline}
\label{hXrelation}
\lim_{\varepsilon\xrightarrow{}0}\mathbb{P}_{\hat{h}^{\varepsilon}_0}(\hat{h}^{\varepsilon}(\mathbf{t}, \mathbf{x}_1)\leq\mathbf{a}_1,\dots, \hat{h}^{\varepsilon}(\mathbf{t}, \mathbf{x}_m)\leq\mathbf{a}_m)
=
\lim_{\varepsilon\rightarrow0}\displaystyle\mathbb{P}_{X^{\varepsilon}_0}
\left(X^{\varepsilon}_{t}(n_1)>a_1,\dots, X^{\varepsilon}_{t}(n_m)>a_m\right),
\end{multline}
where $a_1,\dots, a_m\in\mathbb{R}$ and $t,~n_j,x_j$ are scaled as
\begin{itemize}
\item the discrete time Bernoulli TASEP case
\begin{align}
t=\frac{(2-p)^3}{4p(1-p)}\varepsilon^{-\frac{3}{2}}\mathbf{t},
~
n_i=
\frac{2-p}{4}\varepsilon^{-\frac{3}{2}}\mathbf{t}-\varepsilon^{-1}\mathbf{x}_i
-\frac{1}{2}\varepsilon^{-\frac{1}{2}}\mathbf{a}_i+1,
~
a_i=2\varepsilon^{-1}\mathbf{x}_i-2,
\label{tnascaling}
\end{align}
\item the discrete time geometric TASEP case
\begin{align}
t=\frac{(2-\alpha)^3}{4\alpha(1-\alpha)}\varepsilon^{-\frac{3}{2}}\mathbf{t},
~
n_i=
\frac{2-\alpha}{4(1-\alpha)}\varepsilon^{-\frac{3}{2}}\mathbf{t}-\varepsilon^{-1}\mathbf{x}_i
-\frac{1}{2}\varepsilon^{-\frac{1}{2}}\mathbf{a}_i+1,
~
a_i=2\varepsilon^{-1}\mathbf{x}_i-2.
\label{geotnascaling}
\end{align}
\end{itemize}

Thus we find that our goal, LHS of~\eqref{hXrelation}, can be obtained by taking the $\varepsilon\rightarrow 0$
limit of the expression~\eqref{pro} in Theorem~\ref{Main} under the scaling~\eqref{tnascaling} or~\eqref{geotnascaling}.
The critical step of this problem is the following propositions about the pointwise convergences.
First, we state the result for the discrete time Bernoulli TASEP.
\begin{proposition}(Pointwise convergence for the discrete time Bernoulli TASEP).
\label{scaling}
Under the scaling (\ref{tnascaling}),(dropping the $i$ subscripts) and assuming that (\ref{varx}) holds, if we set $z=\frac{p(2-p)}{4(1-p)}\varepsilon^{-\frac{3}{2}}\mathbf{t}+2\varepsilon^{-1}\mathbf{x}+\varepsilon^{-\frac{1}{2}}(u+\mathbf{a})-2$ and $y'=\varepsilon^{-\frac{1}{2}}v$, then we have for $\mathbf{t}>0$ as $\varepsilon\xrightarrow{} 0$,
\begin{align}
\label{A1}
&\mathbf{S}^{\varepsilon}_{-t,x}(v,u):=\varepsilon^{-\frac{1}{2}}S^{\rm Ber}_{-t,-n}(y',z)\xrightarrow{}\mathbf{S}_{-\mathbf{t,x}}(v,u) 
\\
\label{A2}
&\bar{\mathbf{S}}^{\varepsilon}_{-t,-x}(v,u):=\varepsilon^{-\frac{1}{2}}\bar{S}^{\rm Ber}_{-t,n}(y',z)\xrightarrow{}\mathbf{S}_{-\mathbf{t},-\mathbf{x}}(v,u)
\\
\label{A3}
&\bar{\mathbf{S}}^{\varepsilon, {\rm epi}(-h^{\varepsilon,-}_0)}_{-t,-x}(v,u):=\varepsilon^{-\frac{1}{2}}\bar{S}^{{\rm Ber},{\rm epi}(X_0)}_{-t,n}(y',z)\xrightarrow{}\mathbf{S}^{{\rm epi}(-\hat{h}^{-}_0)}_{-\mathbf{t},-\mathbf{x}}(v,u) 
\end{align}
pointwise, where $\hat{h}^{-}_0(x)=\hat{h}_0(-x)$ for $x\geq0$,
$\mathbf{S}_{\mathbf{t}, \mathbf{x}}(v, u)$ is given by~\eqref{function} and for $\hat{g}\in {\rm LC}$,
$$\mathbf{S}^{{\rm epi}(\hat{g})}_{\mathbf{t},\mathbf{x}}(v,u)=\mathbb{E}_{B(0)=v}[\mathbf{S}_{\mathbf{t}, \mathbf{x}-\bm{\tau}'}(B(\bm{\tau}'),u)\1_{\bm{\tau^{'}}<\infty}]$$
and
\begin{align}
\label{bern}
&S^{\rm Ber}_{-t,-n}(z_1,z_2)=\frac{1}{2\pi i}\oint_{\Gamma_0}dw
\frac{(1-w)^n}{2^{z_2-z_1}w^{n+1+z_2-z_1}}\left(1+\frac{2p}{2-p}\left(w-\frac{1}{2}\right)\right)^t,
\\
\label{cern}
&\bar{S}^{\rm Ber}_{-t,n}(z_1,z_2)=\frac{1}{2\pi i}\oint_{\Gamma_0}dw\frac{(1-w)^{z_2-z_1+n-1}}{2^{z_1-z_2}w^{n}}\left(1-\frac{2p}{2-p}\left(w-\frac{1}{2}\right)\right)^{-t}.
\\
\label{dern}
&\bar{S}^{{\rm Ber},{\rm epi}(X_0)}_{-t,n}(z_1, z_2)=\mathbb{E}_{RW_0=z_1}\left[\bar{S}^{\rm Ber}_{-t, n-\tau}(RW_{\tau}, z_2)\1_{\tau<n}\right]
\end{align}
with $\Gamma_0$ being a simple counterclockwise loop around $0$ not enclosing $1$, $1/p$ and $(1-p)/p$.
\end{proposition}
Next, we state that the point wise convergence for the discrete time geometric TASEP is obtained as the following.
\begin{proposition}(Pointwise convergence for the discrete time geometric TASEP).
\label{gscaling}
Under the scaling (\ref{geotnascaling}),(dropping the $i$ subscripts) and assuming that (\ref{varx}) holds in {\rm LC}, if we set $z=-\frac{\alpha(2-\alpha)}{4(1-\alpha)}\varepsilon^{-\frac{3}{2}}\mathbf{t}+2\varepsilon^{-1}\mathbf{x}+\varepsilon^{-\frac{1}{2}}(u+\mathbf{a})-2$ and $y'=\varepsilon^{-\frac{1}{2}}v$, then we have for $\mathbf{t}>0$ as $\varepsilon\xrightarrow{} 0$,
\begin{align}
\label{gA1}
&\mathbf{S}^{\varepsilon}_{-t,x}(v,u):=\varepsilon^{-\frac{1}{2}}S^{\rm geo}_{-t,-n}(y',z)\xrightarrow{}\mathbf{S}_{-\mathbf{t,x}}(v,u) 
\\
\label{gA2}
&\bar{\mathbf{S}}^{\varepsilon}_{-t,-x}(v,u):=\varepsilon^{-\frac{1}{2}}\bar{S}^{\rm geo}_{-t,n}(y',z)\xrightarrow{}\mathbf{S}_{-\mathbf{t},-\mathbf{x}}(v,u) 
\\
\label{gA3}
&\bar{\mathbf{S}}^{\varepsilon, {\rm epi}(-h^{\varepsilon,-}_0)}_{-t,-x}(v,u):=\varepsilon^{-\frac{1}{2}}\bar{S}^{{\rm geo},{\rm epi}(X_0)}_{-t,n}(y',z)\xrightarrow{}\mathbf{S}^{{\rm epi}(-\hat{h}^{-}_0)}_{-\mathbf{t},-\mathbf{x}}(v,u) 
\end{align}
pointwise,  where $\hat{h}^{-}_0(x)=\hat{h}_0(-x)$ for $x\geq0$,
$\mathbf{S}_{\mathbf{t}, \mathbf{x}}(v, u)$ is given by~\eqref{function} and for $\hat{g}\in {\rm LC}$,
$$\mathbf{S}^{{\rm epi}(\hat{g})}_{\mathbf{t},\mathbf{x}}(v,u)=\mathbb{E}_{B(0)=v}[\mathbf{S}_{\mathbf{t}, \mathbf{x}-\bm{\tau}'}(B(\bm{\tau}'),u)\1_{\bm{\tau^{'}}<\infty}]$$
and
\begin{align}
\label{gbern}
&S^{\rm geo}_{-t,-n}(z_1,z_2)=\frac{1}{2\pi i}\oint_{\Gamma_0}dw\frac{(1-w)^n}{2^{z_2-z_1}w^{n+1+z_2-z_1}}\left(1-\frac{2\alpha}{2-\alpha}\left(w-\frac{1}{2}\right)\right)^{-t},
\\
\label{gcern}
&\bar{S}^{\rm geo}_{-t,n}(z_1,z_2)=\frac{1}{2\pi i}\oint_{\Gamma_0}dw\frac{(1-w)^{z_2-z_1+n-1}}{2^{z_1-z_2}w^{n}}\left(1+\frac{2\alpha}{2-\alpha}\left(w-\frac{1}{2}\right)\right)^{t},
\\
\label{gdern}
&\displaystyle\bar{S}^{{\rm geo},{\rm epi}(X_0)}_{-t,n}(z_1, z_2)=\mathbb{E}_{RW_0=z_1}\left[\bar{S}^{\rm geo}_{-t, n-\tau}(RW_{\tau}, z_2)\1_{\tau<n}\right]
\end{align}
with
$\Gamma_0$ being a simple counterclockwise loop around $0$ not enclosing $1$, $1/\alpha$ and $(1-\alpha)/\alpha$.
\end{proposition}

\section{Distribution function of the TASEP}
\label{s3}
\subsection{Transition probabilities}
\label{s31}
Let 
\begin{equation*}
\Omega_N=\{\vec{x}=(x_N,x_{N-1},\cdots,x_1)\in\Z^N: x_N<\dots<x_2<x_1\}
\end{equation*}
be the Weyl chamber, whose elements express the 
particle positions of the TASEPs.

The main object of this subsection is the transition probability 
of  the TASEP: For $\vec{x},\vec{y}\in\Omega_N$, we define
\begin{equation}
G_t(x_N, \dots, x_1)=\mathbb{P}(X_t=\vec{x}|X_0=\vec{y}),
\end{equation}
which means the probability that at time $t$ the particles are at positions $x_N<\dots<x_2<x_1$ provided that initially they are at 
positions $y_N<\dots<y_2<y_1$.

For all the three types of the TASEPs introduced in Sec.~\ref{cTASEP}-\ref{aTASEP}, the transition probabilities are obtained using Bethe ansatz (See \cite{Schutz}) and represented as determinants.

First, we give the result of the continuous time TASEP introduced in
Sec.~\ref{cTASEP}.
\begin{lemma}{\rm (\cite{Schutz})}\\
\label{MR1}
For the continuous time TASEP with $N\in\{1,2,3,\dots\}$ particles 
and rate $\g\ge 0$ introduced in Sec.\ref{cTASEP}, the 
transition probability has the following determinantal form
\begin{equation}
G^{(\gamma)}_t(x_N,\dots, x_1)=\det[F_{i-j}^{(\gamma)}(x_{N+1-i}-y_{N+1-j})]_{1\leq i,j\leq N}
\end{equation}
with  
\begin{equation}
F_n^{(\gamma)}(x,t)=\frac{(-1)^n}{2\pi i}\oint_{\Gamma_{0,1}}dw\frac{(1-w)^{-n}}{w^{x-n+1}}e^{\gamma t (w-1)}
\end{equation}
where $\Gamma_{0,1}$ is any simple loop oriented anticlockwise which includes $w=0$ and $w=1$.
\end{lemma}

Next we introduce the result on the discrete time Bernoulli TASEP as 
follows.
\begin{lemma}
\label{MR2}
For the discrete time Bernoulli TASEP with $N\in\{1,2,\dots\}$ particles
and parameters $\b_i\ge 0,~i=1,2,\dots,t$ introduced in Sec.~\ref{bTASEP} , the transition probability has the following determinantal form
\begin{equation}
\label{bTASEPgreen}
G^{(\bm{\b})}_t(x_N, \dots, x_1)=\det[F_{i-j}^{(\bm{\b})}(x_{N+1-i}-y_{N+1-j},t)]_{1\leq i,j\leq N}
\end{equation}
with  
\begin{equation}
F_n^{(\bm{\b})}(x,t)=\frac{(-1)^n}{2\pi i}\oint_{\Gamma_{0,1}}dw\frac{(1-w)^{-n}}{w^{x-n+1}}\prod_{j=1}^{t}\frac{1+\beta_j w}{1+\b_j}
\end{equation}
where $\Gamma_{0,1}$ is any simple loop oriented anticlockwise which includes $w=0$ and $w=1$.
\end{lemma}
\begin{proof}
This determinantal formula has been obtained for the time homogeneous 
case $\displaystyle\beta_1=\beta_2=\cdots=\frac{p}{1-p}$ in \cite{Borodin} and~\cite{Rakos}.
If we confirm that the following two equations hold, one easily find that the result can be extended to the time inhomogeneous
case:
\begin{equation}
F^{(\bm{\b})}_n(x, t+1)=\frac{1}{1+\b_{t+1}}F^{(\bm{\b})}_n(x, t)+\frac{\b_{t+1}}{1+\b_{t+1}}F^{(\bm{\b})}_n(x-1, t)
\end{equation}
and
\begin{equation}
F^{(\bm{\b})}_{n-1}(x, t)=F^{(\bm{\b})}_n(x, t)-F^{(\bm{\b})}_n(x+1, t).
\end{equation}
It is easy to see that the above two equations hold.
\end{proof}
We also give the result on  the discrete time geometric TASEP 
introduced in Sec.\ref{aTASEP}.
\begin{lemma}
\label{MR3}
For the discrete time geometric TASEP with $N\in\{1,2,\dots\}$ particles
and parameters $0\le \a_i\le 1,~i=1,2,\dots,t$ introduced in Sec.~\ref{aTASEP} , the transition probability has the following determinantal form
\begin{equation}
\label{ff1}
G^{(\bm{\a})}_t(x_N,\dots,x_1)=\det[F_{i-j}^{(\bm{\a})}(x_{N+1-i}-y_{N+1-j}, t)]_{1\leq i,j \leq N}
\end{equation}
with
\begin{equation*}
\displaystyle F_n^{(\bm{\a})}(x, t)=\frac{(-1)^n}{2\pi i}\oint_{\Gamma_{0,1}} dw \frac{(1-w)^{-n}}{w^{x-n+1}}\prod_{j=1}^{t}\frac{1-\alpha_j}{1-\alpha_j w}
\end{equation*}
where $\Gamma_{0,1}$ is any simple loop oriented anticlockwise which includes $w=0$ and $w=1$.
\end{lemma}
\begin{proof}
We will check that the determinantal representation~\eqref{ff1} satisfies
the Kolmogorov forward equation
\begin{align}
\label{kfeqn}
&~G^{(\bm{\a})}_{t+1}(x_N,\dots,x_1)
=
\hspace{-3mm}
\sum_{\mu\subset\{1, \dots, N-1\}}
\hspace{-3mm}(1-\a_{t+1})^{|\bar{\mu}|+1}
\prod_{i\in\overline{\mu}\cup\{N\}}\sum_{a_i=0}^{k_i-2}\a^{a_i}_{t+1}
\cdot
\prod_{j\in\mu}\a^{k_j-1}_{t+1}
\cdot
G^{(\bm{\a})}_t \left(\vec{x}^{(\mu)}\right)
\end{align} 
where $\mu$ can take the empty set $\phi$, 
$\bar{\mu}:=\{1,\dots, N-1\}\setminus\mu$,
$|\bar{\mu}|$ means the number of elements in $\bar{\mu}$,
and we define $k_i$
and
$\vec{x}^{(\mu)}:=(x^{\mu}_N, \dots, x^{\mu}_1)$ by
\begin{equation}
k_i=
\begin{cases}
x_i-x_{i+1} & \text{for $i=1,\dots, N-1$},
\\
\infty & \text{for $i=N$},
\end{cases}
~~
x^{\mu}_i=
\begin{cases}
x_{i+1}+1 & \text{for $i\in\mu$,}\\
x_i-a_i & \text{for $i\in\overline{\mu}\cup\{N\}$.}
\end{cases}
\end{equation}
RHS in~\eqref{kfeqn} consists of $2^{N-1}$ terms and each element $j$ in the
subset $\mu$ represents the label of the particle which is on $x_{j+1}+1$ at time $t$.
Taking the hopping probability~\eqref{geo} in the geometric TASEP into account,
we see that when  $j\in\mu$, we should assign the weight 
$\a_{t+1}^{\sharp\text{jump}}$ without the factor $1-\a_{t+1}$
for the jump of the $j+1$th particle.  
Thus for the jumps of the $N-|\mu|=|\bar{\mu}|+1$ particles, 
we put the factor $(1-\a_{t+1})^{|\bar{\mu}|+1}$. 
In Appendix \ref{appendixA}, we explain~\eqref{kfeqn} in the case of $N=3$.
We see that~\eqref{kfeqn} is equivalent to the following two conditions
\begin{align}
\label{ff2}
&\displaystyle G^{(\bm{\a})}_{t+1}(x_N, \dots, x_1)
=
\sum_{a_1, \dots, a_n \in\{0, \dots, \infty\}} (1-\alpha_{t+1})^N \alpha_{t+1}^{a_1+ \dots+a_N} G^{(\bm{\a})}_t(x_N-a_N, \dots, x_1-a_1)
\\
\label{ffb}
&\sum_{m, n=0}^{\infty}(1-\alpha_{t+1})\alpha_{t+1}^{m+n}G^{(\bm{\a})}_t(x_N, \dots, x_{k}-m-1, x_{k}-n, x_{k-1}, \dots, x_1)
\notag
\\
&\hspace{5cm}
=\sum_{m=0}^{\infty}\alpha_{t+1}^mG^{(\bm{\a})}_t(x_N, \dots, x_{k}-m-1, x_{k}, x_{k-1}, \dots, x_1)
\end{align}
for $k=1,\dots, N-1$.
In Appendix \ref{appendixB} we show that~\eqref{ff2} and~\eqref{ffb} imply~\eqref{kfeqn}.
Now we will check~\eqref{ff2} and~\eqref{ffb}.
For convenience, we put $F^j_n(x, t)=F^{(\bm{\a})}_n(x-y_{N+1-j}, t)$.
Inserting (\ref{ff1}) into RHS of (\ref{ff2}) 
and using the multilinearity of the determinant, we find that
RHS of (\ref{ff2}) becomes
\begin{equation}
\label{rhsrw}
\begin{split}
&~\sum_{a_1, \dots, a_n \in\{0, \dots, \infty\}} (1-\alpha_{t+1})^N \alpha_{t+1}^{a_1+ \dots+a_N} \det[F^j_{i-j}(x_{N+1-i}-a_{N+1-j}, t)]_{1\leq i,j \leq N}\\
&=\det\left[(1-\alpha_{t+1})\sum_{a_{N+1-i}=0}^{\infty} \alpha_{t+1}^{a_{N+1-i}}F^j_{i-j}(x_{N+1-i}-a_{N+1-j}, t)\right]_{1\leq i,j \leq N}.
\end{split}
\end{equation}
Thus we see that if the functions $F^{(\bm{\a})}_n$ satisfies
\begin{equation}
\label{ff3}
\displaystyle F^{(\bm{\a})}_n(x, t+1)=\sum_{y=0}^{\infty} \alpha_{t+1}^y(1-\alpha_{t+1}) F^{(\bm{\a})}_n(x-y, t),
\end{equation}
then~\eqref{rhsrw} is equal to LHS of~\eqref{ff2}
$\displaystyle G^{(\bm{\a})}_{t+1}(x_N, \dots, x_1)=\det\left[F^j_{i-j}(x_{N+1-i}, t+1)\right]_{1\leq i,j \leq N}$,

We also consider the condition (\ref{ffb}). It can be written as 
\begin{equation}
\label{ff4}
0=\det
\begin{bmatrix}
\vdots\\
\frac{1}{1-\alpha_{t+1}}F^j_{N-k-j}(x_{k}-1, t+1)\\
F^j_{N+1-k-j}(x_{k}, t+1)-F^j_{N+1-k-j}(x_{k}, t)\\
\vdots
\end{bmatrix}_{1\leq j \leq N.}
\end{equation}
One easily sees that it holds if
the functions $F^{(\bm{\a})}_n$ satisfy
\begin{equation}
\label{ff5}
F^{(\bm{\a})}_{n-1}(x-1, t+1)=c(F^{(\bm{\a})}_n(x, t+1)-F^{(\bm{\a})}_n(x, t))
\end{equation}
for arbitrary $c$. Here we choose $c=(1-\alpha_{t+1})/\alpha_{t+1}$.

Therefore the function $F^{(\bm{\a})}_n$ are determined by the two relations (\ref{ff3}) and (\ref{ff5}), as well as the initial condition
\begin{equation}
\label{ff6}
G^{(\bm{\a})}_0(x_N, \dots, x_1)=\delta_{y_N, x_N} \cdots \delta_{y_1, x_1}.
\end{equation}
$F^{(\bm{\a})}_0(x, t)$ is already determined by one particle configurations.
In fact, in this case, $G^{(\bm{\a})}_t(x)=\mathbb{P}(x(t)=x | x(0)=y)=F^{(\bm{\a})}_0(x-y, t)$.
Therefore
\begin{equation}
\displaystyle F^{(\bm{\a})}_0(x-y, t)=\frac{1}{2\pi i} \oint_{\Gamma_0} dw \frac{1}{w^{x-y+1}} \prod_{j=0}^{t}\frac{1-\alpha_j}{1-\alpha_j w}
\end{equation}
where $\Gamma_0$ is any simple loop around $0$ oriented anticlockwise.
This result is consistent with (\ref{ff3}) and (\ref{ff6}).
Denote by $\Delta$ the discrete derivative $\displaystyle\Delta_{\alpha_{t}} f(x, t):=\frac{1-\alpha_t}{\alpha_t}(f(x+1, t)-f(x+1, t-1))$.
Then by (\ref{ff5}),
\begin{equation}
F^{(\bm{\a})}_{-n}(x, t)=(-1)^n(\Delta_{\alpha_t}^n F^{(\bm{\a})}_0)(x, t)
\end{equation}
holds.
Therefore to obtain $F^{(\bm{\a})}_{-n}$ we simply apply 
\begin{equation}
\displaystyle \Delta_{\alpha_t}^n \frac{1}{w^x}\prod_{j=0}^{t}\frac{1-\alpha_j}{1-\alpha_j w}=(-1)^n\frac{(1-w)^n}{w^{x+n}}\prod_{j=0}^{t}\frac{1-\alpha_j}{1-\alpha_j w}.
\end{equation}
From the above, for $n\geq 0$,
\begin{equation}
\displaystyle F^{(\bm{\a})}_{-n}(x, t)=\frac{(-1)^n}{2\pi i}\oint_{\Gamma_0} dw \frac{(1-w)^n}{w^{x+n+1}}\prod_{j=0}^{t}\frac{1-\alpha_j}{1-\alpha_j w}.
\end{equation}
In this case, there is no pole at $w=1$, and therefore replacing $\Gamma_0$ by $\Gamma_{0,1}$ leaves the result unchanged.

For $n>0$, $F^{(\bm{\a})}_n$ is determined by the recurrence relation
\begin{equation}
\label{ff7}
\displaystyle F^{(\bm{\a})}_{n+1}(x, t)=\sum_{y\geq x} F^{(\bm{\a})}_n(y, t)
\end{equation}
together with the property that $F^{(\bm{\a})}_0(x, t)=0$ for $x$ large enough.

In order for (\ref{ff7}) to be satisfied for all $n$, we need to take the poles both at $0$ and $1$.
\end{proof}

Finally, combining the above three formulas in Lemmas~\ref{MR1}--\ref{MR3},
we obtain the transition probability of the TASEP$_{\bm{\a,\b},\g}$ introduced
in Sec.~\ref{abcTASEP}.
\begin{proposition}
\label{abcabc}
For the {\rm TASEP}$_{\bm{\a,\b},\g}$ with $N\in\{1,2,\dots\}$ particles
and parameters ${\bm \a}_{t_1}:=(\a_1,\dots,\a_{t_1})\in~[0,1]^{t_1}$, ${\bm \b}_{t_2}:=(\b_{t_1+1},\dots,\b_{t_1+t_2})\in\R_{\ge 0}^{t_2}$,
$\g>0$ and $t_3>0$ introduced in Sec.~\ref{abcTASEP} , the transition probability to $t=t_1+t_2+t_3$ 
has the following determinantal form
\begin{equation}
G^{\a, \b, \g}_{t}(x_N,\dots, x_1)=\det[F_{i-j}^{\a,\b,\g}(x_{N+1-i}-y_{N+1-j}, t)]_{1\leq i,j \leq N}
\end{equation}
with
\begin{equation*}
\displaystyle F^{\a,\b,\g}_n(x, t)=\frac{(-1)^n}{2\pi i}\oint_{\Gamma_{0,1}} dw \frac{(1-w)^{-n}}{w^{x-n+1}}f_{\alpha, \beta, \gamma}(w, t)
\end{equation*}
where $\Gamma_{0,1}$ is any simple loop oriented anticlockwise which includes $w=0$ and $w=1$ and 
\begin{align}
\label{fabc}
f_{\alpha, \beta, \gamma}(w, t)
=
 \prod_{j=1}^{t_1}\frac{1-\a_j}{1-\a_j w}\cdot\prod_{j=t_1+1}^{t_1+t_2}\frac{1+\b_j w}{1+\b_j}\cdot e^{\g t_3(w-1)}.
\end{align}
\end{proposition}
\begin{proof}
We can obviously prove from Lemma \ref{MR1}, Lemma \ref{MR2}, and Lemma \ref{MR3}.
\end{proof}
\begin{remark}

In the case of the step initial condition, $y_j=-j,~j=1,2,\dots$,
it has been known that the TASEP has a connection to the Schur 
measures and processes~\cite{IS2007,Johansson,Alisa,Jon}. It is natural to
ask the corresponding Schur measure to the TASEP$_{\bm{\a,\b},\g}$ with 
the step initial condition. Combining the findings in~\cite{Alisa,Jon}, we expect
that the position of the $N$th particle from the right is equivalent in distribution to the marginal $\l_N$ of the 
Schur measure,
\begin{align}
s_{\l}(\underbrace{1,1, \dots, 1}_{N \text{times}})s_{\l}(\rho)/Z, 
\end{align}
where $\l=(\l_1,\dots,\l_N)\in\Z_{\ge 0}$ with $\l_1\ge\dots\ge \l_N$
is a partition, $Z$ is the normalization constant,  $s_\l(x_1,\dots,x_N)$ is the Schur symmetric polynomial and $s_\l(\rho)$
is the Schur function with the Schur positive specialization $\rho$ defined by the relation
of the specialization of the complete symmetric functions $h_k,k=0,1,2\dots$
\begin{align}
\sum_{z=0}^{\i} z^k h_k(\rho)
=
\prod_{i=1}^{t_1}\frac{1}{1-z \a_i}\cdot \prod_{j=t_1+1}^{t_1+t_2}(1+z \b_j)\cdot
e^{\g t_3 z}.
\end{align}
\end{remark}
\subsection{Biorthogonal ensembles for the joint distribution functions}
\label{s32}
In the following we consider the joint distribution function of the particle positions
in the TASEP$_{\a,\b,\g}$ introduced in Sec.~\ref{abcTASEP}.
We will give a formula in terms of a Fredholm determinant whose kernel 
can be written in an explicit form.
\begin{theorem}
\label{stu}
We consider the $\rm{TASEP}_{\a,\b,\g}$ introduced in Sec.~\ref{abcTASEP}.
For $(n_1,n_2,\dots,n_m)\in \Z^m$ with $1\leq n_1<n_2<\cdots<n_m\le N$ and $(a_1,a_2,\dots,a_m)\in \Z^m$,
we have
\begin{equation}
\mathbb{P}(X_t(n_j)>a_j, j=1,\dots,m)=\mathrm{det}(I-\bar{\chi}_{\bm a} K_t \bar{\chi}_{\bm a})_{\ell^2(\{n_1,\dots,n_m\}\times\mathbb{Z})}.
\end{equation}
Here the right hand side is a Fredholm determinant with the kernel
\begin{equation}
\label{kernel1}
\displaystyle K_t(n_i, x_i; n_j, x_j)=-Q^{n_j-n_i}(x_i, x_j)\1_{n_i<n_j}+\sum_{k=1}^{n_j}\Psi^{n_i}_{n_i-k}(x_i)\Phi^{n_j}_{n_j-k}(x_j)
\end{equation}
where $Q^n(x_i,x_j)$ represents $n$-times convolution of  
$Q(x,y)=1/2^{x-y}\cdot \1_{x>y}$. The functions $\Psi^n_k(x)$
and $\Phi^n_k(x), k=0,\dots, n-1$ are defined as follows:
For $\Psi^n_k(x)$ with $k\le n-1$, we define
\begin{equation}
\displaystyle\Psi^n_k(x):=\frac{1}{2\pi i}\oint_{\Gamma_0}dw
\frac{(1-w)^k}{2^{x-X_0(n-k)}w^{x+k+1-X_0(n-k)}}
f_{\alpha, \beta, \gamma}(w, t)
\end{equation}
where $\Gamma_0$ is any positively oriented simple loop including the pole at $w=0$ and $f_{\alpha, \beta, \gamma}(w, t)$ is defined by~\eqref{fabc}.
The functions $\Phi^n_k(x), k=0,\dots, n-1$, are defined implicitly by\\
\leftline{(1) The biorthogonality relation $\displaystyle\sum_{x\in\mathbb{Z}}\Psi^n_k(x)\Phi^n_l(x)=\1_{k=l}$;}
\leftline{(2) $2^{-x}\Phi^n_k(x)$ is a polynomial of degree at most $n-1$ in $x$ for each $k$.}
\end{theorem}
\begin{remark}
The fact that the joint distribution of particle positions can be expressed by the Fredholm determinant is proved by \cite{Borodin} for the discrete time Bernoulli TASEP and by \cite{Sasamoto} for the continuous time TASEP.
The above result includes a generalization of the initial conditions for particle position in the results of distribution of particle position of \cite{Borodin} and \cite{Sasamoto}, and is the result when the continuous time TASEP, the discrete time Bernoulli TASEP and the discrete time geometric TASEP are mixed. 
\end{remark}
\begin{proof}
This proof can be proved in the same way as Theorem 4.3 in \cite{Quastel} by using the propositions and lemmas written in Chapter 4 of \cite{Quastel}.
Therefore, only the outline of the proof is described below. (See \cite{Quastel} for more details.)

From the proof of Proposition 3.2 in \cite{Borodin} and the proof of Theorem 2.1 in \cite{Sasamoto}, we can found that if the following three equations are satisfied, the proof can be done regardless of the form of $f_{\alpha, \beta, \gamma}(w, t)$:
\begin{equation*}
\displaystyle F_{n+1}(x, t)=\sum_{y\geq x} F_n(y, t),
\end{equation*}
\begin{equation*}
\Psi^n_{k}(x)=\frac{(-1)^{k}}{2^{x-X_0(n-k)}}F_{-k}(x-X_0(n-k), t),  \ \ \ \ \ \ \ \ \ \ k=0,\dots, n
\end{equation*}
\begin{equation}
\label{QQQ}
Q^{n-m}\Psi^n_{n-k}=\Psi^m_{m-k}.
\end{equation}
Therefore, it is sufficient for us to check the above three equations, but it is not hard to confirm that the above three equations hold.

This completes the proof.
\end{proof}
In the following, we will write $\Phi^n_k(x)$ that was not explicitly written in previous research \cite{Borodin} in an explicit form.
 
First, we prepare the tools to use.
$Q^m$ can easily be taken from definition $Q$;
\begin{equation}
\label{qex}
\displaystyle Q^m(x,y)=\frac{1}{2^{x-y}}\binom{x-y-1}{m-1}\1_{x\geq y+m}.
\end{equation}
As operators on $\ell^2(\mathbb{Z})$, $Q$ and $Q^m$ are invertible;
\begin{equation}
\label{in}
Q^{-1}(x,y)=2\cdot\1_{x=y-1}-\1_{x=y}, \ \  Q^{-m}(x,y)=(-1)^{y-x+m}2^{y-x}\binom{m}{y-x}.
\end{equation}
Now we define
\begin{equation}
\label{R}
\displaystyle R_{\alpha, \beta, \gamma, t}(x, y):=\frac{1}{2\pi i}\oint_{\Gamma_0}dw\frac{\displaystyle f_{\alpha, \beta, \gamma}(w, t)}{2^{x-y}w^{x-y+1}},
\end{equation}
where $$f_{\alpha, \beta, \gamma}(w, t)=\prod_{j=1}^{t_1}\frac{1-\a_j}{1-\a_j w}\cdot\prod_{j=t_1+1}^{t_1+t_2}\frac{1+\b_j w}{1+\b_j}\cdot e^{\g t_3(w-1)}.$$
Note that $\displaystyle\Psi^n_0=R_{\alpha, \beta, \gamma, t}\delta_{X_0(n)}$ with $\delta_y(x)=\1_{x=y}$.

Then, the following lemma holds.
\begin{lemma}
For $n\in\mathbb{Z}$,
\begin{equation}
\label{psi}
\displaystyle\Psi^n_k=R_{\alpha, \beta, \gamma, t}Q^{-k}\delta_{X_0(n-k)}.
\end{equation}
\end{lemma}
\begin{proof}
By (\ref{QQQ}) and (\ref{R}),
$$\displaystyle\Psi^n_k=Q^{-k}R_{\alpha, \beta, \gamma, t}\delta_{X_0(n-k)}$$
holds.

Now, note that $Q$ and $R_{\alpha, \beta, \gamma, t}$ commute, because the kernels $Q(x,y)$ and $R_{\alpha, \beta, \gamma, t}(x, y)$ only depend on $x-y$. Therefore, we obtain
\begin{equation*}
\displaystyle\Psi^n_k=R_{\alpha, \beta, \gamma, t}Q^{-k}\delta_{X_0(n-k)}.
\end{equation*}
\end{proof}
From the expression of $R_{\alpha, \beta, \gamma, t}$, we define
\begin{equation}
\label{R1}
\displaystyle R^{-1}_{\alpha, \beta, \gamma, t}(x,y):=\frac{1}{2\pi i}\oint_{\Gamma_0}dw\frac{\displaystyle f^{-1}_{\alpha, \beta, \gamma}(w, t)}{2^{x-y}w^{x-y+1}},
\end{equation}
It is not hard to check that $R_{\alpha, \beta, \gamma, t}R^{-1}_{\alpha, \beta, \gamma, t}=R^{-1}_{\alpha, \beta, \gamma, t}R_{\alpha, \beta, \gamma, t}=I$.
At this time, the following theorem holds.
\begin{theorem}
\label{ta2}
Fix $0\leq k<n$ and consider particles at $X_0(1)>X_0(2)>\cdots>X_0(n)$.\\
Let $h^n_k(l,z)$ be the unique solution to the initial-boundary value problem for the backwards heat equation
\begin{subnumcases}
{}(Q^{*})^{-1}h^n_k(l,z)=h^n_k(l+1,z) & \text{$l<k, z\in\mathbb{Z}$},\label{a1}\\
h^n_k(k,z)=2^{z-X_0(n-k)} & \text{$z\in\mathbb{Z}$},\label{b1}\\
h^n_k(l,X_0(n-l))=0 & \text{$l<k$}.\label{c1}
\end{subnumcases}
Then the functions $\Phi^n_k$ from Theorem \ref{stu} are given by
\begin{equation}
\label{yss}
\displaystyle\Phi^n_k(z)=(R^{*}_{\alpha, \beta, \gamma, t})^{-1}h^n_k(0,\cdot)(z)=\sum_{y\in\mathbb{Z}}h^n_k(0,y)R^{-1}_{\alpha, \beta, \gamma, t}(y,z).
\end{equation}
Here $Q^{*}(x,y)=Q(y,x)$ is the kernel of the adjoint of $Q$ (and likewise for $R^{*}_{\alpha, \beta, \gamma, t}$).
\end{theorem}
\begin{remark}
It is not true that in general $Q^{*}h^n_k(l+1,z)=h^n_k(l,z)$.
In fact, $Q^{*}h^n_k(k,z)$ is divergent from the following.
\begin{equation*}
\begin{split}
\displaystyle Q^{*}h^n_k(k,\cdot)(z)&=\sum_{y\in\mathbb{Z}}h^n_k(k,y)Q(y,z)\\
&=\sum_{y\in\mathbb{Z}}2^{y-X_0(n-k)}\frac{1}{2^{y-z}}\1_{y>z}\\
&=\sum_{y\in\mathbb{Z}, y>z}2^{z-X_0(n-k)}\\
&=\infty.
\end{split}
\end{equation*}
\end{remark}
This proof will be given in Appendix \ref{appendixC}, which is 
almost the same as \cite{Matetski}. 


\subsection{Representation of the TASEP kernel in terms of a hitting probability}
\label{s33}
Combining~Theorem \ref{stu} with (\ref{psi}) and (\ref{yss}), we have obtained 
the following expression of the kernel $K_t$~\eqref{kernel1},
\begin{equation}
\label{kernelkt}
\displaystyle K_t(n_i,\cdot ;n_j,\cdot)
=
-Q^{n_j-n_i}\1_{n_i<n_j}
+
R_{\alpha, \beta, \gamma, t}Q^{-n_i}G_{0,n_j}R^{-1}_{\alpha, \beta, \gamma, t}.
\end{equation}
Here $Q$, $R_{\alpha, \beta, \gamma, t}$ and $R_{\alpha, \beta, \gamma, t}^{-1}$ are given by~\eqref{qex},~\eqref{R}
and \eqref{R1} respectively and  $G_{0,n_j}$ is defined by
\begin{equation}
\label{g0n}
\displaystyle G_{0,n}(z_1, z_2)=\sum_{k=0}^{n-1}Q^{n-k}(z_1, X_0(n-k))h^n_k(0,z_2),
\end{equation}
where $h^n_k$ is the solution of (\ref{a1})-(\ref{c1}).

In this subsection, following the method in \cite{Matetski}, we further rewrite the kernel
in order to take the KPZ scaling limit. We use the fact that $h^n_k$ can be written as hitting probabilities of random walk.
Let $RW^{*}_m$ with $RW^{*}_{-1}=c$ be the position of the random walk
with Geom$[\frac{1}{2}]$ jumps strictly to the right starting from $c\in\Z$, i.e.
\begin{align*}
RW^{*}_m=c+\chi_0+\cdots+\chi_m,
\end{align*}
where $\chi_j,~j=0,1,\dots,m$ are i.i.d. random variables with $\P(\chi_j=k)=1/2^{k+1},~k\in\Z_{\ge 0}$. Note that $Q^*$ defined 
below~\eqref{yss} represents the transition kernel of the random walk:
for $m=-1,0,\dots$, we have
\begin{align}
\label{qsex}
Q^{*}(x,y)=\P(RW^{*}_{m+1}=x|RW^{*}_m=y). 
\end{align}
For $0\leq l\leq k\leq n-1$ we define the stopping times
\begin{equation}
\label{tau}
\tau^{l,n}=\min\{m\in\{l,\dots,n-1\}: RW^{*}_m>X_0(n-m)\},
\end{equation}
where we set $\min\varnothing=\infty$.

Then, we have the following.
\begin{lemma}{\rm(\cite{Matetski})}\\
\label{nobo}
For $z\leq X_0(n-l)$, the function $h^n_k$ can be written by
\begin{equation}
\label{kljp}
h^n_k(l,z)=\mathbb{P}_{RW^{*}_{l-1}=z}(\tau^{l,n}=k)
\end{equation}
which is the probability of the walk starting at $z\in\mathbb{Z}$ at time $l-1\in\mathbb{Z}$ and hitting $X_0(n-k)$ at time $k\in\mathbb{Z}$.
\end{lemma}
\begin{remark}
Eq. (\ref{kljp}) is written in \cite{Matetski} and the proof is left to the readers as Exercise 5.17 in \cite{Quastel}.
In Appendix \ref{appendixD} we give an answer.
\end{remark}
From the memoryless property of geometric distribution we get for all $y>X_0(n-k)$,
\begin{equation}
\label{mlgeo}
\mathbb{P}_{RW^{*}_{-1}=z}(\tau^{0,n}=k,RW^{*}_k=y)=2^{X_0(n-k)-y}\mathbb{P}_{RW^{*}_{-1}=z}(\tau^{0,n}=k)
\end{equation}
and as a consequence we get for $z_2\leq X_0(n)$, $G_{0,n}(z_1, z_2)$ can be expressed as
\begin{equation}
\begin{split}
\label{grw}
\displaystyle G_{0,n}(z_1, z_2)&=\sum_{k=0}^{n-1}\mathbb{P}_{RW^{*}_{-1}=z_2}(\tau^{0,n}=k)(Q^{*})^{n-k}(X_0(n-k), z_1)\\
&=\sum_{k=0}^{n-1}\sum_{z>X_0(n-k)}\mathbb{P}_{RW^{*}_{-1}=z_2}(\tau^{0,n}=k, RW^{*}_k=z)(Q^{*})^{n-k-1}(z, z_1)\\
&=\mathbb{P}_{RW^{*}_{-1}=z_2}(\tau^{0,n}<n, RW^{*}_{n-1}=z_1),
\end{split}
\end{equation}
where in the second equality we used ~\eqref{qex} and ~\eqref{mlgeo}.
while in the third one we used~\eqref{qsex}.
Note that RHS of the above equation represents the probability for the walk starting at $z_2\in\mathbb{Z}$ at time $-1$ to end up at $z_1\in\mathbb{Z}$ after $n$ steps, having hit the curve $(X_0(n-m))_{m=0,\dots,n-1}$ in between.

The next step is to extend the region $z_2\leq X_0(n)$ in~\eqref{grw} to $z_2\in\Z$.
We begin by observing that for each fixed $y_1$ and $n\geq 1$, $2^{-y_2}Q^n(y_1, y_2)$ extends in $y_2$ to a polynomial $2^{-y_2}\bar{Q}^{(n)}(y_1, y_2)$ of degree $n-1$ with
\begin{equation}
\label{vq}
\displaystyle\bar{Q}^{(n)}(y_1, y_2)=\frac{1}{2\pi i}\oint_{\Gamma_0}dw \frac{(1+w)^{y_1-y_2-1}}{2^{y_1-y_2}w^n}.
\end{equation}
Now, for $y_1-y_2\geq 1$, we note that
\begin{equation}
\label{eqqqb}
\bar{Q}^{(n)}(y_1, y_2)=Q^n(y_1, y_2).
\end{equation}
By (\ref{in}) and (\ref{vq}), for $n>1$, we get
\begin{equation}
Q^{-1}\bar{Q}^{(n)}=\bar{Q}^{(n)}Q^{-1}=\bar{Q}^{(n-1)}.
\end{equation}
Also, we get 
\begin{equation}
Q^{-1}\bar{Q}^{(1)}=\bar{Q}^{(1)}Q^{-1}=0.
\end{equation}
\begin{remark}
We note that 
\begin{equation*}
\begin{split}
\displaystyle\bar{Q}^{(n)}\bar{Q}^{(m)}(x,y)&=\sum_{z\in\mathbb{Z}}\bar{Q}^{(n)}(x,z)\bar{Q}^{(m)}(z,y)\\
&=\sum_{z\in\mathbb{Z}}\frac{1}{2\pi i}\oint_{\Gamma_0}dw \frac{(1+w)^{x-z-1}}{2^{x-z}w^n}\cdot\frac{1}{2\pi i}\oint_{\Gamma_0}dw \frac{(1+w)^{z-y-1}}{2^{z-y}w^n}\\
&=\infty.
\end{split}
\end{equation*}
\end{remark}
Using the extension of $Q^m$, we have the following lemma.
\begin{lemma}{\rm (\cite{Matetski})}\\
\label{le}
For all $z_1, z_2\in\mathbb{Z}$, we have
\begin{equation}
\label{lemg0}
G_{0,n}(z_1,z_2)=\mathbb{E}_{RW_0=z_1}\left[\bar{Q}^{(n-\tau)}(RW_{\tau}, z_2) \1_{\tau<n}\right],
\end{equation}
where $RW_m$ and $\t$ are defined by~Definition.~\ref{defRWtau}.
\end{lemma}
Although the proof is given in Lemma 2.4 of \cite{Matetski}, in Appendix \ref{appendixE} we give its
outline for self-containedness.
Thus from~\eqref{kernelkt} and~\eqref{lemg0}, we see that 
the kernel $K_t$~\eqref{kernel1} can be expressed as
\begin{align}
\label{kerktbm}
K_t(n_1,x_1;n_2,x_2)
&=
-Q^{n_2-n_1}(x_1,x_2)\1_{n_1<n_2}
\notag
\\
&~+
\sum_{x,y\in\Z}
(R_{\alpha, \beta, \gamma, t}Q^{-n_1})(x_1,x)
\E_{RW_0=x}
\left[\bar{Q}^{(n_2-\t)}(RW_\t,y)
R^{-1}_{\alpha, \beta, \gamma,t}(y,x_2)\1_{\t<n_2}\right].
\end{align}
\subsection{Formulas for the mixed TASEP with right finite initial data: Proof of 
Theorem~\ref{Main}}
\label{s34}
To show Theorem~\ref{Main}, we have the following relations.
\begin{proposition}
\label{proppp}
\begin{align}
\label{S}
&A^{-1}_{\alpha, \beta, \gamma}(t)(R_{\alpha, \beta, \gamma, t}Q^{-n})^{*}(z_1,z_2)
=S_{-t,-n}(z_1,z_2),
\\
\label{s}
&A_{\alpha, \beta, \gamma}(t) \bar{Q}^{(n)}
R^{-1}_{\alpha, \beta, \gamma, t} (z_1,z_2)
=\bar{S}_{-t,n}(z_1,z_2).
\end{align}
Here $S_{-t,-n}(z_1,z_2)$ and $\bar{S}_{-t,n}(z_1,z_2)$ are defined 
by~\eqref{stn} and \eqref{bstn} respectively and $A_{\alpha, \beta, \gamma}(t)$
is defined by
\begin{align}
A_{\alpha, \beta, \gamma}(t):=e^{-\frac{\gamma t_3}{2}}\prod_{j=1}^{t_1}\frac{1-\alpha_j}{2-\alpha_j}\prod_{j=t_1+1}^{t_1+t_2}\frac{2+\beta_j}{1+\beta_j}.
\end{align}
\end{proposition}
We will give this proof in Appendix \ref{appendixF}.
\begin{proof}[Proof of Theorem \ref{Main}]
First, we consider right finite initial data.
If $X_0(1)<\infty$ then we are in the setting of the above section.
Formula (\ref{Kt}) follow directly from above definition.

Now, we check (\ref{pro}).
To check (\ref{pro}), it is enough to check 
\begin{equation*}
Q^{n_j-n_i}K_t^{(n_j)}=(S_{-t, -n_i})^{*}\bar{S}^{\rm epi(X_0)}_{-t,n_j},
\end{equation*}
where $K_t^{(n_j)}=R_{\alpha, \beta, \gamma, t}Q^{-n_j}G_{0,n_j}R^{-1}_{\alpha, \beta, \gamma, t}$.

Because $Q$ and $R_{\alpha, \beta, \gamma, t}$ commute, by lemma \ref{le},
\begin{equation*}
\begin{split}
\displaystyle Q^{n_j-n_i}K_t^{(n_j)}&=R_{\alpha, \beta, \gamma, t}Q^{-n_i}G_{0,n_j}R^{-1}_{\alpha, \beta, \gamma, t}\\
&=A^{-1}_{\alpha, \beta, \gamma}(t)R_{\alpha, \beta, \gamma, t}Q^{-n_i}G_{0,n_j}R^{-1}_{\alpha, \beta, \gamma, t}A_{\alpha, \beta, \gamma}(t)\\
&=(S_{-t, -n_i})^{*}\bar{S}^{\rm epi(X_0)}_{-t,n_j}.
\end{split}
\end{equation*}
If $X_0(j)=\infty$ for $j=1,\dots l$ and $X_0(l+1)<\infty$ then
\begin{equation*}
\mathbb{P}_{X_0}(X_t(n_j)>a_j, j=1,\dots,M)=\det(I-\bar{\chi}_aK^{(l)}_t\bar{\chi}_a)_{\ell^2(\{n_1,\dots,n_M\}\times\mathbb{Z}}
\end{equation*}
with the correlation kernel
\begin{equation*}
K^{(l)}_t(n_i, \cdot  ; n_j, \cdot)=-Q^{n_j-n_i}\1_{n_i<n_j}+(S_{-t, -n_i})^{*}\bar{S}^{\rm epi(\theta_lX_0)}_{-t,n_j-l},
\end{equation*}
where $\theta_lX_0(j)=X_0(l+j)$.
Now, using the fact that $Q^l\bar{S}^{\rm epi(\theta_lX_0)}_{-t,n_j-l}=\bar{S}^{\rm epi(X_0)}_{-t,n_j}$ and (\ref{S}), we have that (\ref{Kt}) still holds in this case.
\end{proof}
\section{Asymptotics}
\label{asy}

In this section we take the KPZ scaling limit for 
the discrete time Bernoulli and geometric TASEP
and prove Proposition \ref{scaling} and \ref{gscaling}.
\subsection{Proof of Proposition \ref{scaling}}
First, we prove (\ref{A1}).
By changing variables $\displaystyle w=\frac{1}{2}(1-\varepsilon^{\frac{1}{2}}y)$, we have
\begin{equation}
\begin{split}
\label{e}
\displaystyle(\ref{bern})&=\frac{1}{2\pi i}\oint_{C_{\varepsilon}}\frac{1}{2}\varepsilon^{\frac{1}{2}}dy\frac{\{\frac{1}{2}(1+\varepsilon^{\frac{1}{2}}y)\}^n}{2^{z-y'}\{\frac{1}{2}(1-\varepsilon^{\frac{1}{2}}y)\}^{n+1+z-y'}}\left(1-\frac{p}{2-p}\varepsilon^{\frac{1}{2}}y\right)^t\\
&=\frac{1}{2\pi i}\oint_{C_{\varepsilon}}\varepsilon^{\frac{1}{2}}dy\frac{(1+\varepsilon^{\frac{1}{2}}y)^n}{(1-\varepsilon^{\frac{1}{2}}y)^{n+1+z-y'}}\left(1-\frac{p}{2-p}\varepsilon^{\frac{1}{2}}y\right)^t
\end{split}
\end{equation}
where $C_{\varepsilon}$ is a circle of radius $\varepsilon^{-\frac{1}{2}}$ centred at $\varepsilon^{-\frac{1}{2}}$.
In order to apply the saddle point method, we rewrite~\eqref{e} as
\begin{equation}
\label{fF3}
\frac{1}{2\pi i}\oint_{C_{\varepsilon}}\varepsilon^{\frac{1}{2}} \ e^{f(\varepsilon^{\frac{1}{2}}y)+\varepsilon^{-1}F_2(\varepsilon^{\frac{1}{2}}y)+\varepsilon^{-\frac{1}{2}}F_1(\varepsilon^{\frac{1}{2}}y)+F_0(\varepsilon^{\frac{1}{2}}y)}dy,
\end{equation}
where the functions $f(x)$ and $F_i(x),~i=0,1,2$ are defined by
\begin{align}
&f(x)=\frac{2-p}{4}\hat{t}\log(1+x)-\frac{2-p}{4(1-p)}
\hat{t}\log(1-x)+\frac{(2-p)^3}{4p(1-p)}\hat{t}\log\left(1-\frac{p}{2-p}x\right),
\\
&F_2(x)=-\mathbf{x}\log(1-x^2),~F_1(x)=(v-u-\frac{1}{2}\mathbf{a})\log(1-x)-\frac{1}{2}\mathbf{a}\log(1+x),~ F_0(x):=\log 2(1+x)
\end{align}
with $\hat{t}:=\varepsilon^{-\frac{3}{2}}\mathbf{t}$.
Calculating the derivatives of $f(x)$ up to the third order, we have
\begin{align}
&f'(x)=\frac{x^2}{(1-x^2)(1-\frac{p}{2-p}x)}\hat{t},
~f^{''}(x)=\frac{2x-\frac{p}{2-p}x^2-\frac{p}{2-p}x^4}{(1-\frac{p}{2-p}x-x^2+\frac{p}{2-p}x^3)^2}\hat{t},
\notag
\\
&f^{(3)}(x)=\frac{2\left(1+3x^2-8\frac{p}{2-p}x^3+3\left(\frac{p}{2-p}\right)^2x^4+\left(\frac{p}{2-p}\right)^2x^6\right)}{(1-\frac{p}{2-p}x-x^2+\frac{p}{2-p}x^3)^3} \hat{t}.
\end{align} 
Thus we see that $f(x)$ has the double saddle point
at $x=0$,
\begin{equation}
\label{g}
 f(0)=0, \ f'(0)=0, \ f^{''}(0)=0 \ {\rm and} \ f^{(3)}(0)=2\hat{t}.
\end{equation}
Therefore, for small $\varepsilon$, $f(x)$ is expanded as
\begin{equation}
\label{fex}
f(\varepsilon^{\frac{1}{2}}y)\approx\frac{\mathbf{t}}{3}y^3.
\end{equation}
For small $\varepsilon$, we also have
\begin{equation}
\label{Fasy}
\varepsilon^{-1}F_2(\varepsilon^{\frac{1}{2}}y)\approx\mathbf{x}y^2, \ \varepsilon^{-\frac{1}{2}}F_1(\varepsilon^{\frac{1}{2}}y)\approx(u-v)y, \ F_0(\varepsilon^{\frac{1}{2}}y)\approx\log 2.
\end{equation}
Now, we see the convergence of the integration path.
First, we deform $C_{\varepsilon}$ to the contour $\langle_{\varepsilon} \  \cup \  C^{\frac{\pi}{3}}_{\varepsilon}$ where $\langle_{\varepsilon}$ is the part of Airy contour $\langle$ within the ball of radius $\varepsilon^{-\frac{1}{2}}$ centred at $\varepsilon^{-\frac{1}{2}}$, and $C^{\frac{\pi}{3}}_{\varepsilon}$ is the part of $C_{\varepsilon}$ to the right of $\langle$. From~\eqref{fF3},~\eqref{fex}, and~\eqref{Fasy}, we have
\begin{align}
\label{eairy}
&~\lim_{\e\rightarrow 0}
\frac{1}{2\pi i}\int_{\langle_\varepsilon}\varepsilon^{\frac{1}{2}} \ e^{f(\varepsilon^{\frac{1}{2}}y)+\varepsilon^{-1}F_2(\varepsilon^{\frac{1}{2}}y)+\varepsilon^{-\frac{1}{2}}F_1(\varepsilon^{\frac{1}{2}}y)+F_0(\varepsilon^{\frac{1}{2}}y)}dy
=\mathbf{S}_{\mathbf{t}, \mathbf{x}}(y),
\end{align}
where $\mathbf{S}_{\mathbf{t}, \mathbf{x}}(y)$ is defined by~\eqref{function}.
Thus the remaining part is to show that the integral over $C^{\frac{\pi}{3}}_{\varepsilon}$ converges to $0$.
To see this note that the real part of the exponent of the integral over $C_{\varepsilon}$ in (\ref{e}), parametrized as $y=\varepsilon^{-\frac{1}{2}}(1-e^{i\theta})$, is given by 
\begin{equation*}
\displaystyle\varepsilon^{-\frac{3}{2}}\mathbf{t}\biggl[\frac{(2-p)^3}{8p(1-p)}\log\left(1+\frac{4p(1-p)}{(2-p)^2}(\cos\theta-1)\right)+\left(\frac{2-p}{8}+\mathcal{O}(\varepsilon^{\frac{1}{2}})\right)\log(5-4\cos\theta)\biggr].
\end{equation*}
Because the $y\in C^{\frac{\pi}{3}}_{\varepsilon}$ correspond to $\frac{\pi}{3}<|\theta|\leq\pi$\footnote{\label{footnote1}Since $\theta=0$ corresponds to the origin $\mathbf{0}\in C_{\varepsilon}$ and $\langle$ is the positively oriented contour going the straight lines from $e^{-\frac{i\pi}{3}}\infty$ to $e^{\frac{i\pi}{3}}\infty$ through 0, the domain of $\theta$ can be written by this domain.}, using $\log(1+x)< x$\footnote{\label{footnote2}This inequality comes from $\log(1+x)\leq x$ for $x>-1$, but since $x=0$ corresponds to $\theta=0$ in (\ref{im1}) and (\ref{im2}), we use this inequality to correspond to the calculations of (\ref{im1}) and (\ref{im2}).} for $x\in(-1, \infty)\setminus\{0\}$, we get 
\begin{equation}
\label{im1}
\displaystyle\varepsilon^{-\frac{3}{2}}\mathbf{t}\left[\frac{(2-p)^3}{8p(1-p)}\log\left(1+\frac{4p(1-p)}{(2-p)^2}(\cos\theta-1)\right)\right]< \frac{2-p}{2} \varepsilon^{-\frac{3}{2}}\mathbf{t}\left[\cos\theta-1\right]
\end{equation}
and
\begin{equation}
\label{im2}
\displaystyle\varepsilon^{-\frac{3}{2}}\mathbf{t}\left[\frac{2-p}{8}\log(5-4\cos\theta)\right]<\frac{2-p}{2}\varepsilon^{-\frac{3}{2}}\mathbf{t}\left[1-\cos\theta\right].
\end{equation}
Therefore, for sufficiently small $\varepsilon$, the exponent there is less than $-\varepsilon^{-\frac{3}{2}}\kappa\mathbf{t}$ for some $\kappa>0$.
Hence we see that the part $C_{\e}^{\frac{\pi}{3}}$ of the integral 
vanishes and this completes the proof of~\eqref{A1}.
We can also prove(\ref{A2}) in the similar way to~\eqref{A1}
thus omit the proof.

For the proof of~\eqref{A3},  we define the scaled walk $\displaystyle\mathbf{B}^{\varepsilon}(x)=\varepsilon^{\frac{1}{2}}(\text{RW}_{\varepsilon^{-1}x}+2\varepsilon^{-1}x-1)$ for $x\in\varepsilon\mathbb{Z}_{\geq0}$, interpolated linearly in between, and let $\bm\tau^{\varepsilon}$ be the hitting time by $\mathbf{B}^{\varepsilon}$ of ${\rm epi}(-\hat{h}^{\varepsilon}(0,\cdot)^{-})$,
where $\hat{h}^{\varepsilon}({\bf t, x})$ is defined by~\eqref{Hei}
and $\hat{h}^{\varepsilon}({\bf t, x})^-=\hat{h}^{\varepsilon}({\bf t, -x})$.
By Donsker's invariance principle \cite{Patrick}, $\displaystyle\mathbf{B}^{\varepsilon}(x)$ converges locally uniformly in distribution to a Brownian motion $\mathbf{B}(x)$ with diffusion coefficient $2$.
Combining this with (\ref{varx}), one finds the hitting time $\bm\tau^{\varepsilon}$ converges to $\bm\tau$. (For more detailed proof,
see Proposition 3.2 in \cite{Matetski}).)
This leads to (\ref{A3}).
\subsection{Proof of Proposition \ref{gscaling}}
Proposition \ref{gscaling} can be shown in a similar manner to 
Proposition~\ref{scaling}. Here we give only the proof of~\eqref{gA1}.
\eqref{gA2} can be obtained in a parallel way to~\eqref{gA1} whereas
\eqref{gA3} follows from~\eqref{gA2} and the Donsker\rq{}s invariance principle
as in the case of~\eqref{A3} in Proposition~\ref{scaling}.
As for~\eqref{e} and~\eqref{fF3}, we rewrite~\eqref{gbern}
by changing variables $\displaystyle w=(1-\varepsilon^{\frac{1}{2}}y)/2$, 
\begin{align}
\label{ggge}
\displaystyle(\ref{gbern})
&=\frac{1}{2\pi i}\oint_{C_{\varepsilon}}\varepsilon^{\frac{1}{2}}dy\frac{(1+\varepsilon^{\frac{1}{2}}y)^n}{(1-\varepsilon^{\frac{1}{2}}y)^{n+1+z-y}}\left(1+\frac{\alpha}{2-\alpha}\varepsilon^{\frac{1}{2}}y\right)^{-t}
\notag
\\
&=\frac{1}{2\pi i}\oint_{C_{\varepsilon}}\varepsilon^{\frac{1}{2}} \ e^{g(\varepsilon^{\frac{1}{2}}y)+\varepsilon^{-1}G_2(\varepsilon^{\frac{1}{2}}y)+\varepsilon^{-\frac{1}{2}}G_1(\varepsilon^{\frac{1}{2}}y)+G_0(\varepsilon^{\frac{1}{2}}y)}dy
\end{align}
where $C_{\varepsilon}$ is a circle of radius $\varepsilon^{-\frac{1}{2}}$ centred at $\varepsilon^{-\frac{1}{2}}$ and $g(x)$, $G_j(x),~j=0,1,2$ are defined by
\begin{align}
&g(x)=\frac{2-\alpha}{4(1-\alpha)}\hat{t}\log(1+x)-\frac{2-\alpha}{4}\hat{t}\log(1-x)-\frac{(2-\alpha)^3}{4\alpha(1-\alpha)}\hat{t}\log\left(1+\frac{\alpha}{2-\alpha}x\right)
\notag
\\
&G_2(x)=-\mathbf{x}\log(1-x^2),
~G_1(x)=(v-u-\frac{1}{2}\mathbf{a})\log(1-x)-\frac{1}{2}\mathbf{a}\log(1+x),
~G_0(x)=\log 2(1+x)
\end{align}
with $\hat{t}:=\varepsilon^{-\frac{3}{2}}\mathbf{t}$.
Here we apply the saddle point method to~\eqref{ggge}. Noting
\begin{align}
&g'(x)=\frac{x^2}{(1-x^2)(1+\frac{\alpha}{2-\alpha}x)}\hat{t},
~
g^{''}(x)=\frac{2x+\frac{\alpha}{2-\alpha}x^2+\frac{\alpha}{2-\alpha}x^4}{(1+\frac{\alpha}{2-\alpha}x-x^2-\frac{\alpha}{2-\alpha}x^3)^2}\hat{t},
\notag
\\
&g^{(3)}(x)=\frac{2\left(1+3x^2+8\frac{\alpha}{2-\alpha}x^3+3\left(\frac{\alpha}{2-\alpha}\right)^2x^4+\left(\frac{\alpha}{2-\alpha}\right)^2x^6\right)}{(1+\frac{\alpha}{2-\alpha}x-x^2-\frac{\alpha}{2-\alpha}x^3)^3} \hat{t},
\end{align}
we find $g(x)$ has a double saddle point at $x=0$,
\begin{align}
\label{gfex}
g(0)=0, \ g'(0)=0, \ g^{''}(0)=0 \ {\rm and} \ g^{(3)}(0)=2\hat{t}.
\end{align}
Therefore, for small $\varepsilon$, we have
\begin{equation}
g(\varepsilon^{\frac{1}{2}}y)\approx\frac{\mathbf{t}}{3}y^3.
\end{equation}
For $G_i(x),~i=0,1,2$, we easily see
\begin{equation}
\label{gFex}
\varepsilon^{-1}G_2(\varepsilon^{\frac{1}{2}}y)\approx\mathbf{x}y^2, \ \varepsilon^{-\frac{1}{2}}G_1(\varepsilon^{\frac{1}{2}}y)\approx(u-v)y, \ G_0(\varepsilon^{\frac{1}{2}}y)\approx\log 2.
\end{equation}
As discussed above~\eqref{eairy}, we divide the contour 
$C_{\varepsilon}$ in (\ref{ggge})  into two parts 
$\langle_{\varepsilon}~\cup~C^{\frac{\pi}{3}}_{\varepsilon}$.
From \eqref{ggge}, \eqref{gfex}, and~\eqref{gFex}, we have
\begin{align}
\label{geairy}
&~\lim_{\e\rightarrow 0}
\frac{1}{2\pi i}\int_{\langle_\varepsilon}\varepsilon^{\frac{1}{2}} \ e^{g(\varepsilon^{\frac{1}{2}}y)+\varepsilon^{-1}G_2(\varepsilon^{\frac{1}{2}}y)+\varepsilon^{-\frac{1}{2}}G_1(\varepsilon^{\frac{1}{2}}y)+G_0(\varepsilon^{\frac{1}{2}}y)}dy
=\mathbf{S}_{\mathbf{t}, \mathbf{x}}(y),
\end{align}
where $\mathbf{S}_{\mathbf{t}, \mathbf{x}}(y)$ is defined by~\eqref{function}.

Finally we show that the part coming from $C^{\frac{\pi}{3}}_{\varepsilon}$
vanishes as $\varepsilon\rightarrow0$
To see this note that the real part of the exponent of the integral over $C_{\varepsilon}$ in (\ref{ggge}), parametrized as $y=\varepsilon^{-\frac{1}{2}}(1-e^{i\theta})$, is given by 
\begin{equation*}
\begin{split}
\displaystyle\varepsilon^{-\frac{3}{2}}\mathbf{t}\biggl[\left(\frac{2-\alpha}{8(1-\alpha)}+\mathcal{O}(\varepsilon^{\frac{1}{2}})\right)&\log\left(1+\frac{4(4-\alpha)(1-\alpha)}{(2-\alpha)^2+4\alpha(1-\cos\theta)}(1-\cos\theta)\right)\\
&+\left(\frac{(2-\alpha)(4-\alpha)}{8\alpha}+\mathcal{O}(\varepsilon^{\frac{1}{2}})\right)\log\left(1+\frac{4\alpha}{(2-\alpha)^2+4\alpha(1-\cos\theta)}(\cos\theta-1)\right)\biggr].
\end{split}
\end{equation*}
Note that we used an expression transform
\begin{equation*}
\begin{split}
&\frac{2-\alpha}{8(1-\alpha)}\log(5-4\cos\theta)-\frac{(2-\alpha)^3}{8\alpha(1-\alpha)}\log\left(1+\frac{4\alpha(1-\cos\theta)}{(2-\alpha)^2}\right)\\
&=\frac{2-\alpha}{8(1-\alpha)}\left[\log(5-4\cos\theta)-\log\left(1+\frac{4\alpha(1-\cos\theta)}{(2-\alpha)^2}\right)\right]-\frac{(2-\alpha)(4-\alpha)}{8\alpha}\log\left(1+\frac{4\alpha(1-\cos\theta)}{(2-\alpha)^2}\right)\\
&=\frac{2-\alpha}{8(1-\alpha)}\log\left(1+\frac{4(4-\alpha)(1-\alpha)}{(2-\alpha)^2+4\alpha(1-\cos\theta)}(1-\cos\theta)\right)\\
&+\frac{(2-\alpha)(4-\alpha)}{8\alpha}\log\left(1+\frac{4\alpha}{(2-\alpha)^2+4\alpha(1-\cos\theta)}(\cos\theta-1)\right).
\end{split}
\end{equation*}
Because the $y\in C^{\frac{\pi}{3}}_{\varepsilon}$ correspond to $\frac{\pi}{3}<|\theta|\leq\pi$, using $\log(1+x)< x$ for $x\in(-1, \infty)\setminus\{0\}$(See \ref{footnote1} and \ref{footnote2}), we get 
\begin{equation*}
\displaystyle\varepsilon^{-\frac{3}{2}}\mathbf{t}\left[\frac{2-\alpha}{8(1-\alpha)}\log\left(1+\frac{4(4-\alpha)(1-\alpha)}{(2-\alpha)^2+4\alpha(1-\cos\theta)}(1-\cos\theta)\right)\right]< \frac{(2-\alpha)(4-\alpha)}{2\left\{(2-\alpha)^2+4\alpha(1-\cos\theta)\right\}}\varepsilon^{-\frac{3}{2}}\mathbf{t}\left[1-\cos\theta\right]
\end{equation*}
and
\begin{equation*}
\displaystyle\varepsilon^{-\frac{3}{2}}\mathbf{t}\left[\frac{(2-\alpha)(4-\alpha)}{8\alpha}\log\left(1+\frac{4\alpha}{(2-\alpha)^2+4\alpha(1-\cos\theta)}(\cos\theta-1)\right)\right]< \frac{(2-\alpha)(4-\alpha)}{2\left\{(2-\alpha)^2+4\alpha(1-\cos\theta)\right\}}\varepsilon^{-\frac{3}{2}}\mathbf{t}\left[\cos\theta-1\right].
\end{equation*}
Therefore, for sufficiently small $\varepsilon$, the exponent there is less than $-\varepsilon^{-\frac{3}{2}}\kappa\mathbf{t}$ for some $\kappa>0$.
Hence this part of the integral vanishes.
\subsection{Proof of Theorem \ref{special}}
\label{sspf}
By using Propositions \ref{scaling} or \ref{gscaling}, we can prove Theorem \ref{special} as following.
This proof is almost the same as \cite{Matetski}.
First, we change variables in the kernel as in Proposition \ref{scaling} (resp. Proposition \ref{gscaling}), so that for $z_i=\frac{p(2-p)}{4(1-p)}\varepsilon^{-\frac{3}{2}}\mathbf{t}+2\varepsilon^{-1}\mathbf{x}_i+\varepsilon^{-\frac{1}{2}}(u_i+\mathbf{a}_i)-2$ (resp. $z_i=-\frac{\alpha(2-\alpha)}{4(1-\alpha)}\varepsilon^{-\frac{3}{2}}\mathbf{t}+2\varepsilon^{-1}\mathbf{x}_i+\varepsilon^{-\frac{1}{2}}(u_i+\mathbf{a}_i)-2$) we need to compute the limit of $\varepsilon^{-\frac{1}{2}}(\bar{\chi}_{2\varepsilon^{-1}\mathbf{x}-2}K_t \bar{\chi}_{2\varepsilon^{-1}\mathbf{x}-2})(z_i, z_j)$.
Note that the change of variables turns $\bar{\chi}_{2\varepsilon^{-1}\mathbf{x}-2}(z)$ into $\bar{\chi}_{-\mathbf{a}}(u)$.
We have $n_i<n_j$ for small $\varepsilon$ if and only if $\mathbf{x}_j<\mathbf{x}_i$ and in this case we have, under our scaling,
\begin{equation}
\varepsilon^{-\frac{1}{2}}Q^{n_j-n_i}(z_i,z_j)\xrightarrow{}e^{(\mathbf{x}_i-\mathbf{x}_j)\partial^2}(u_i,u_j),
\end{equation}
as $\varepsilon\xrightarrow{}0$.
For the second term in (\ref{Kt}), by Proposition \ref{scaling} we get
\begin{equation}
\begin{split}
\varepsilon^{-\frac{1}{2}}(S_{-t, -n_i})^{*}\bar{S}^{\rm epi(X_0)}_{-t,n_j}(z_i,z_j)&=\varepsilon^{-\frac{1}{2}}\int_{-\infty}^{\infty}d\nu(S_{-t, -n_i})^{*}(z_i,\nu)\bar{S}^{\rm epi(X_0)}_{-t,n_j}(\nu,z_j)\\
&=\varepsilon^{-1}\int_{-\infty}^{\infty}d\nu(S_{-t, -n_i})^{*}(z_i,\varepsilon^{-\frac{1}{2}}\nu)\bar{S}^{\rm epi(X_0)}_{-t,n_j}(\varepsilon^{-\frac{1}{2}}\nu,z_j)\\
&=\int_{-\infty}^{\infty}d\nu (\mathbf{S}^{\varepsilon}_{-t,x_i})^{*}(u_i,\nu)\bar{\mathbf{S}}^{\varepsilon, \rm{epi}(-h^{\varepsilon,-}_0)}_{-t,-x_j}(\nu,u_j)\\
&=(\mathbf{S}^{\varepsilon}_{-t,x_i})^{*}\bar{\mathbf{S}}^{\varepsilon, \rm{epi}(-h^{\varepsilon,-}_0)}_{-t,-x_j}(u_i,u_j)\\
&\xrightarrow[\varepsilon\xrightarrow{}0]{}(\mathbf{S}_{-t,x_i})^{*}\mathbf{S}^{\rm{epi}(-\hat{h}^{-}_0)}_{-t,-x_j}(u_i,u_j).
\end{split}
\end{equation}
Therefore, we have a limiting kernel
\begin{equation}
\mathbf{K}_{\lim}(x_i, u_i ; x_j, u_j)=-e^{(\mathbf{x}_i-\mathbf{x}_j)\partial^2}(u_i,u_j)\1_{x_i>x_j}+(\mathbf{S}_{-t,x_i})^{*}\mathbf{S}^{\rm{epi}(-\hat{h}^{-}_0)}_{-t,-x_j}(u_i,u_j)
\end{equation}
surrounded by projection $\bar{\chi}_{-\mathbf{a}}$.
It is nicer to have projection $\chi_{\mathbf{a}}$, so we change variables $u_i\mapsto-u_i$ and replace the Fredholm determinant of the kernel by that of its adjoint to get $\det\left(\mathbf{I}-\chi_{\mathbf{a}}K^{{\rm hypo}(\hat{h}_0)}_{\mathbf{t}, {\rm ext}}\chi_{\mathbf{a}}\right)$ with $K^{{\rm hypo}(\hat{h}_0)}_{\mathbf{t}, {\rm ext}}(u_i, u_j)=\mathbf{K}_{\lim}(x_j, -u_j ; x_i, -u_i)$.

By using $\textstyle(\mathbf{S}_{\mathbf{t}, \mathbf{x}})^{*}\mathbf{S}_{\mathbf{t}, \mathbf{-x}}=I$ and $\textstyle\mathbf{S}^{{\rm epi}(\hat{h})}_{-\mathbf{t},\mathbf{x}}(v,u)=\mathbf{S}^{{\rm hypo}(-\hat{h})}_{\mathbf{t},\mathbf{x}}(-v,-u)$ (see \cite{Matetski} for more information on these equations), we get the following:
\begin{equation*}
\mathbf{K}^{{\rm hypo}(\hat{h}_0)}_{\mathbf{t},{\rm ext}}(\mathbf{x}_i, \cdot ;\mathbf{x}_j, \cdot)=-e^{(\mathbf{x}_j-\mathbf{x}_i)\partial^2}\1_{\mathbf{x}_i<\mathbf{x}_j}+\left(\mathbf{S}^{{\rm hypo}(\hat{h}^{-}_0)}_{\mathbf{t},-\mathbf{x}_i}\right)^{*}\mathbf{S}_{\mathbf{t}, \mathbf{x}_j}.
\end{equation*}
\appendix
\section{The Kolmogorov forward equation for the discrete time geometric 
TASEP with $N=3$}
\label{appendixA}
Here, we explain~\eqref{kfeqn} in more detail in the case of  $N=3$. 
In this case~\eqref{kfeqn} can be decomposed into four terms,
\begin{equation}
\label{kge}
G^{(\bm{\a})}_{t+1}(x_3,x_2,x_1)=
G^{(\bm{\a}, 1)}_t(x_3,x_2,x_1)+G^{(\bm{\a}, 2)}_t(x_3,x_2,x_1)
+G^{(\bm{\a}, 3)}_t(x_3,x_2,x_1)+G^{(\bm{\a}, 4)}_t(x_3,x_2,x_1)
\end{equation}
where for $k_1:=x_1-x_2$,~$k_2:=x_2-x_3$, and
\begin{align}
\label{bulkeq1}
&G^{(\bm{\a}, 1)}_t(x_3,x_2,x_1)
=
\sum_{a_3=0}^{\infty}\sum_{a_2=0}^{k_2-2}\sum_{a_1=0}^{k_1-2}(1-\alpha_{t+1})^3\alpha_{t+1}^{a_1+a_2+a_3}G^{(\bm{\a})}_t(x_3-a_3, x_2-a_2, x_1-a_1)
\\
\label{bulkeq3}
&G^{(\bm{\a}, 2)}_t(x_3,x_2,x_1)
=
\sum_{a_3=0}^{\infty}\sum_{a_2=0}^{k_2-2}(1-\alpha_{t+1})^2\alpha_{t+1}^{a_2+a_3+k_1-1}G^{(\bm{\a})}_t(x_3-a_3, x_2-a_2, x_2+1)
\\
\label{bulkeq2}
&G^{(\bm{\a}, 3)}_t(x_3,x_2,x_1)
=
\sum_{a_3=0}^{\infty}\sum_{a_1=0}^{k_1-2}(1-\alpha_{t+1})^2\alpha_{t+1}^{a_1+a_3+k_2-1}G^{(\bm{\a})}_t(x_3-a_3, x_3+1, x_1-a_1)
\\
\label{bulkeq4}
&G^{(\bm{\a}, 4)}_t(x_3,x_2,x_1)
=
\sum_{a_3=0}^{\infty}(1-\alpha_{t+1})\alpha_{t+1}^{a_3+k_1+k_2-2}G^{(\bm{\a})}_t(x_3-a_3, x_3+1, x_2+1)
\end{align}
The four equations (\ref{bulkeq1}) through (\ref{bulkeq4}) correspond to the 
case $\mu=\phi$, $\mu=\{1\}$, $\mu=\{2\}$, and $\mu=\{1,2\}$ respectively
and the situations for all the equations are illustrated in Fig.~\ref{fig1}(a)-(d) below.
\begin{figure}[H]
\begin{picture}(200,200)
\put(10,170){(a)}
\put(10,120){(b)}
\put(10,70){(c)}
\put(10,20){(d)}
\put(30,150){\includegraphics[keepaspectratio, scale=0.3]{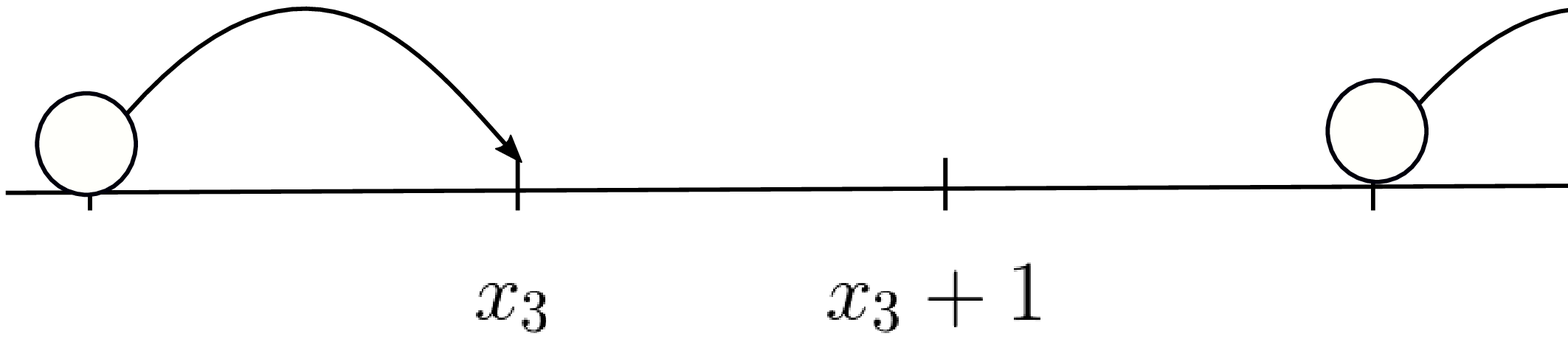}}
\put(30,100){\includegraphics[keepaspectratio, scale=0.3]{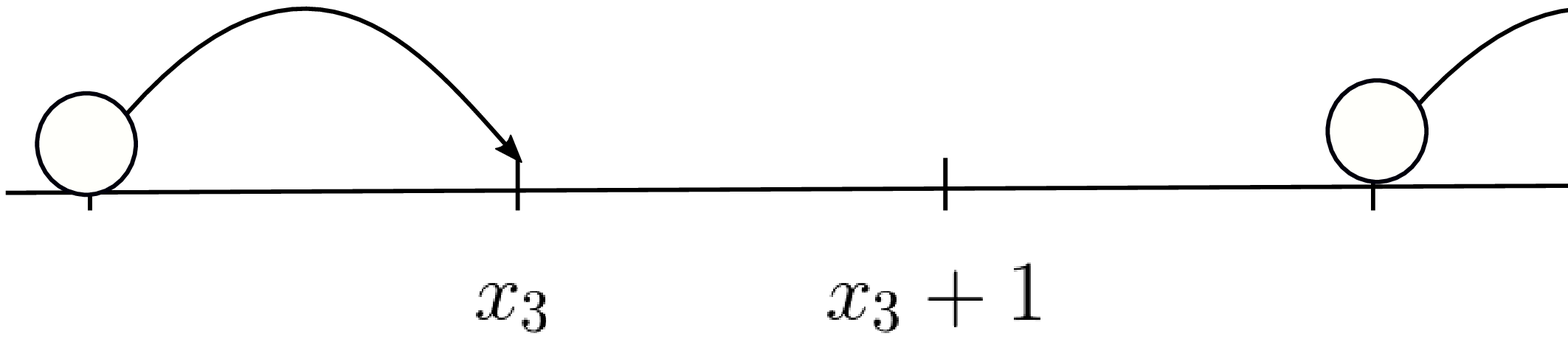}}
\put(30,50){\includegraphics[keepaspectratio, scale=0.3]{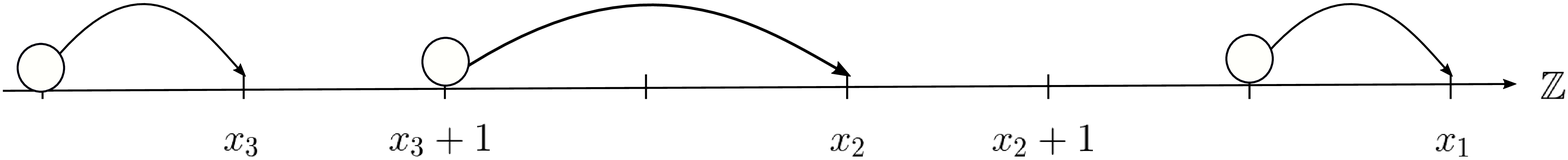}}
\put(30,0){\includegraphics[keepaspectratio, scale=0.3]{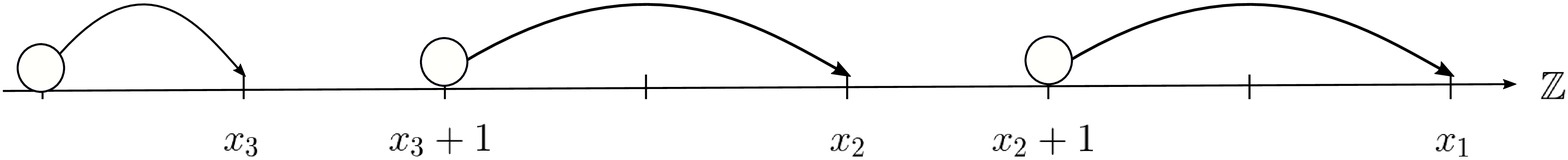}}
\end{picture}
\caption{The evolutions of the geometric TASEP with 3 particles. The circles
correspond to the particles and they move to the directions of arrows during
time step $t\rightarrow t+1$.
(a) The case $\mu=\phi$. Neither of the particles are blocked by each other. 
(b) The case $\mu=\{1\}$. At time $t$, the first particle (from the right) 
is at $x_2+1$ which leads to the blocking of the second particle.  
(c) The case $\mu=\{2\}$. At time $t$, the second particle
is at $x_3+1$ which leads to the blocking of the third particle.  
(d) The case $\mu=\{1,2\}$. At time $t$, the first and second particles
are at $x_2+1$ and $x_3+1$ respectively which leads to the blockings of both 
the second and the third particles.}
\label{fig1}
\end{figure}
\section{On the Kolmogorov forward equation for the discrete time geometric TASEP}
\label{appendixB}
Now we prove the equivalence between the Kolmogorov forward equation~\eqref{kfeqn} and two conditions~\eqref{ff2} and~\eqref{ffb}.
First, we show below the equivalence
\begin{equation}
\begin{split}
\label{Nper}
\displaystyle\prod_{i=1}^{N}&\sum_{a_i=0}^{\infty}(1-\a_{t+1})\a^{a_i}_{t+1}G^{(\bm{\a})}_t(x_N-a_N, \dots, x_1-a_1)\\
&=\sum_{\mu\subset\{1, \dots, N-1\}}
\hspace{-3mm}(1-\a_{t+1})^{|\bar{\mu}|+1}
\prod_{i\in\overline{\mu}\cup\{N\}}\sum_{a_i=0}^{k_i-2}\a^{a_i}_{t+1}
\cdot
\prod_{j\in\mu}\a^{k_j-1}_{t+1}
\cdot
G^{(\bm{\a})}_t \left(\vec{x}^{(\mu)}\right)
\end{split}
\end{equation}
by using the equation (\ref{ffb}) with $k=1$ and the version of $N-1$ particles
in~\eqref{Nper},
\begin{equation}
\begin{split}
\label{nper}
\displaystyle\prod_{i=2}^{N}&\sum_{a_i=0}^{\infty}(1-\a_{t+1})\a^{a_i}_{t+1}G^{(\bm{\a})}_t(x_N-a_N, \dots, x_2-a_2, x_1)\\
&=\sum_{\nu\subset\{2, \dots, N-1\}}
\hspace{-3mm}(1-\a_{t+1})^{|\bar{\nu}|+1}
\prod_{i\in\overline{\nu}\cup\{N\}}\sum_{a_i=0}^{k_i-2}\a^{a_i}_{t+1}
\cdot
\prod_{j\in\nu}\a^{k_j-1}_{t+1}
\cdot
G^{(\bm{\a})}_t \left(\vec{x}^{(\nu)},x_1\right)
\end{split}
\end{equation}
where $\bar{\nu}:=\{2,\dots, N-1\}\setminus\nu$, 
$\vec{x}^{(\nu)}:=(x^{\nu}_N, \dots, x^{\nu}_2)$ 
with
\begin{equation}
x^{\nu}_i=
\begin{cases}
x_{i+1}+1 & \text{for $i\in\nu$,}\\
x_i-a_i & \text{for $i\in\bar{\nu}\cup\{N\}$.}
\end{cases}
\end{equation}
In LHS of~\eqref{Nper}, we divide the sum of $a_1$ as follows.
\begin{equation}
\begin{split}
\label{convper}
\rm LHS \ \rm of \ (\ref{Nper})&=\displaystyle\left(\prod_{i=2}^{N}\sum_{a_i=0}^{\infty}(1-\a_{t+1})\a^{a_i}_{t+1}\right)\Biggl\{\left(\sum_{a_1=0}^{k_1-2}(1-\a_{t+1})\a^{a_1}_{t+1}\right)G^{(\bm{\a})}_t(x_N-a_N, \dots, x_1-a_1)\\
&+\left(\sum_{a_1=k_1-1}^{\infty}(1-\a_{t+1})\a^{a_1}_{t+1}\right)G^{(\bm{\a})}_t(x_N-a_N, \dots, x_1-a_1)\Biggr\}\\
&=\displaystyle\left(\prod_{i=2}^{N}\sum_{a_i=0}^{\infty}(1-\a_{t+1})\a^{a_i}_{t+1}\right)\Biggl\{\left(\sum_{a_1=0}^{k_1-2}(1-\a_{t+1})\a^{a_1}_{t+1}\right)G^{(\bm{\a})}_t(x_N-a_N, \dots, x_1-a_1)\\
&+\left(\sum_{a_1=0}^{\infty}(1-\a_{t+1})\a^{a_1+k_1-1}_{t+1}\right)G^{(\bm{\a})}_t(x_N-a_N, \dots, x_2-a_2, x_2+1-a_1)\Biggr\}.
\end{split}
\end{equation}
By using (\ref{ffb}) with $k=1$, we find
\begin{equation}
\begin{split}
\label{conv2per}
(\ref{convper})=\displaystyle\left(\prod_{i=2}^{N}\sum_{a_i=0}^{\infty}(1-\a_{t+1})\a^{a_i}_{t+1}\right)\Biggl\{&\left(\sum_{a_1=0}^{k_1-2}(1-\a_{t+1})\a^{a_1}_{t+1}\right)G^{(\bm{\a})}_t(x_N-a_N, \dots, x_1-a_1)\\
&+\a^{k_1-1}_{t+1}G^{(\bm{\a})}_t(x_N-a_N, \dots, x_2-a_2, x_2+1)\Biggr\}.
\end{split}
\end{equation}
By applying (\ref{nper}) to (\ref{conv2per}),
\begin{equation*}
\begin{split}
(\ref{conv2per})&=\Biggl\{
\sum_{\nu\subset\{2, \dots, N-1\}}
\hspace{-3mm}(1-\a_{t+1})^{|\bar{\nu}|+1}
\prod_{i\in\overline{\nu}\cup\{N\}}\sum_{a_i=0}^{k_i-2}\a^{a_i}_{t+1}
\cdot
\prod_{j\in\nu}\a^{k_j-1}_{t+1}
\Biggr\}\\
&\times\Biggl\{\left(\sum_{a_1=0}^{k_1-2}(1-\a_{t+1})\a^{a_1}_{t+1}\right)G^{(\bm{\a})}_t(\vec{x}^{(\nu)}, x_1-a_1)+\a^{k_1-1}_{t+1}G^{(\bm{\a})}_t(\vec{x}^{(\nu)}, x_2+1)\Biggr\}\\
&=\rm RHS \ \rm of \ (\ref{Nper}).
\end{split}
\end{equation*}
Thus we have shown~\eqref{Nper} by using~\eqref{nper}.
Similarly, we can show~(\ref{nper}) by using the equation (\ref{ffb}) with 
$k=2$ and 
\begin{equation*}
\begin{split}
\displaystyle\prod_{i=3}^{N}&\sum_{a_i=0}^{\infty}(1-\a_{t+1})\a^{a_i}_{t+1}G^{(\bm{\a})}_t(x_N-a_N, \dots, x_3-a_3, x_2, x_1)\\
&=
\sum_{\l\subset\{3, \dots, N-1\}}
\hspace{-3mm}(1-\a_{t+1})^{|\bar{\l}|+1}
\prod_{i\in\bar{\l}\cup\{N\}}\sum_{a_i=0}^{k_i-2}\a^{a_i}_{t+1}
\cdot
\prod_{j\in\l}\a^{k_j-1}_{t+1}
\cdot
G^{(\bm{\a})}_t \left(\vec{x}^{(\l)},x_2,x_1\right)
\end{split}
\end{equation*}
where $\overline{\lambda}:=\{3,\dots, N-1\}\setminus\lambda$, $\vec{x}^{(\lambda)}:=(x^{\lambda}_N, \dots, x^{\lambda}_3)$ and for $i=3,\dots, N$
\begin{equation*}
x^{\lambda}_i=
\begin{cases}
x_{i+1}+1 & \text{for $i\in\lambda$,}\\
x_i-a_i & \text{for $i\in\overline{\lambda}\cup\{N\}$.}
\end{cases}
\end{equation*}
Therefore, by repeatedly using the similar calculation, we can show 
the equivalence~\eqref{Nper} by using conditions can be obtained from (\ref{ffb}) for $k=1,\dots, N-1$, which leads to the equivalence between the Kolmogorov forward equation~\eqref{kfeqn} and two conditions~\eqref{ff2} and~\eqref{ffb}. 
\section{Proof of Theorem \ref{ta2}}
\label{appendixC}
The existence and uniqueness of solutions of (\ref{a1})-(\ref{c1}) is elementary consequence of the fact that $ker(Q^{*})^{-1}$ has dimension 1 and it is spanned by the function $2^z$, which allows us to march forwards from the initial condition $h^n_k(k,z)=2^{z-X_0(n-k)}$ uniquely solving the boundary value problem $h^n_k(l,X_0(n-k))=0$ at each step.\\
First, we prove that $2^{-x}h^n_k(0,x)$ is a polynomial of degree at most $k$.
We use the mathematical induction.
By (\ref{b1}), 
\begin{equation}
2^{-x}h^n_k(k,x)=2^{-x}2^{x-X_0(n-k)}=2^{-X_0(n-k)}.
\end{equation}
Therefore, $2^{-x}h^n_k(k,x)$ is polynomial of degree 0.

Now, assume that $\hat{h}^n_k(l,x):=2^{-x}h^n_k(l,x)$ is a polynomial of degree at most $k-l$ for some $0<l\leq k$.
By (\ref{in}) and (\ref{a1}),
\begin{equation*}
\begin{split}
\displaystyle\hat{h}^n_k(l,y)&=2^{-y}(Q^{*})^{-1}h^n_k(l-1,y)\\
&=2^{-y}(2\cdot h^n_k(l-1,y-1)-h^n_k(l-1,y))\\
&=2^{-(y-1)}h^n_k(l-1,y-1)-2^{-y}h^n_k(l-1,y)\\
&=\hat{h}^n_k(l-1,y-1)-\hat{h}^n_k(l-1,y).
\end{split}
\end{equation*}
Taking the sum over $x\geq X_0(n-l+1)$, one sees
\begin{equation*}
\begin{split}
\displaystyle\sum_{y=X_0(n-l+1)+1}^{x}2^{-y}h^n_k(l,y)&=\sum_{y=X_0(n-l+1)+1}^{x}\hat{h}^n_k(l,y)\\
&=\sum_{y=X_0(n-l+1)+1}^{x}(\hat{h}^n_k(l-1,y-1)-\hat{h}^n_k(l-1,y))\\
&=\hat{h}^n_k(l-1,X_0(n-l+1))-\hat{h}^n_k(l-1,x).\\
\end{split}
\end{equation*}
Therefore, using (\ref{c1}), we have $\displaystyle\hat{h}^n_k(l-1,x)=-\sum_{y=X_0(n-l+1)+1}^{x}2^{-y}h^n_k(l,y)$.
 
By the induction hypothesis, $\hat{h}^n_k(l-1,x)$ is a polynomial of degree at most $k-l+1$ because $\hat{h}^n_k(l,y)$ is a polynomial of degree at most $k-l$\footnote{This can be understood from Faulhaber's formula : $\displaystyle\sum_{j=1}^{n} j^k=\frac{1}{k+1}\sum_{j=0}^{k}\binom{k+1}{j} B_j n^{k+1-j}$ where $B_j$ is Bernoulli number.}.

Similarly, taking the sum $x<X_0(n-l+1)$, we get $\displaystyle\hat{h}^n_k(l-1,x)=\sum_{y=x+1}^{X_0(n-l+1)}\hat{h}^n_k(l,y)$,which is again a polynomial of degree at most $k-l+1$.
From the above, it was shown that $2^{-x}h^n_k(0,x)$ is a polynomial of degree at most $k$.

Now, we show that $\displaystyle\sum_{y\in\mathbb{Z}}h^n_k(0,y)R^{-1}_{\alpha, \beta, \gamma, t}(y,z)$, which is the rhs of (\ref{yss}), satisfies the condition (2) in Theorem \ref{stu}.
By (\ref{R1}), we have
\begin{equation}
\begin{split}
\displaystyle 2^{-z}\sum_{y\in\mathbb{Z}}h^n_k(0,y)R^{-1}_{\alpha, \beta, \gamma, t}(y,z)&=2^{-z}\sum_{y\geq z}\left(\frac{1}{2\pi i}\oint_{\Gamma_0}dw\frac{\displaystyle f^{-1}_{\alpha, \beta, \gamma}(w, t)}{2^{y-z}w^{y-z+1}}\right)h^n_k(0,y)\\
&=\sum_{y\geq z}\left(\frac{1}{2\pi i}\oint_{\Gamma_0}dw\frac{\displaystyle f^{-1}_{\alpha, \beta, \gamma}(w, t)}{w^{y-z+1}}\right) 2^{-y} h^n_k(0,y)\\
&=\sum_{x \geq 0} \left(\frac{1}{2\pi i}\oint_{\Gamma_0}dw\frac{\displaystyle f^{-1}_{\alpha, \beta, \gamma}(w, t)}{w^{x+1}}\right) 2^{-(x+z)}h^n_k(0,x+z).
\end{split}
\end{equation}
Because  $2^{-z}h^n_k(0,z)$ is a polynomial of degree at most $k$, it is enough to note that the sum is a polynomial of degree at most $k$ in $z$ as well.
Next, we check the biorthogonality relation (1) of Theorem \ref{stu}.
Using (\ref{psi}), we get 
\begin{equation*}
\begin{split}
\displaystyle\sum_{z\in\mathbb{Z}}\Psi^n_l(z)\Phi^n_k(z)&=\sum_{z_1,z_2\in\mathbb{Z}}\sum_{z\in\mathbb{Z}}R_{\alpha, \beta, \gamma, t}(z,z_1)Q^{-l}(z_1,X_0(n-l))h^n_k(0,z_2)R^{-1}_{\alpha, \beta, \gamma, t}(z_2,z)\\
&=\sum_{z\in\mathbb{Z}}Q^{-l}(z,X_0(n-l))h^n_k(0,z)=(Q^{*})^{-l}h^n_k(0,X_0(n-l)),
\end{split}
\end{equation*}
where in the first equality we have used the decay of $R_{\alpha, \beta, \gamma, t}$ and the fact that $2^{-x}h^n_k(0,x)$ is a polynomial together with the fact that the $z_1$ sum is finite to apply Fubini.

For $l\leq k$, from (\ref{b1}) and (\ref{c1}), we have the boundary condition
\begin{equation}
\label{BC}
h^n_k(l,X_0(n-l))=\1_{l=k}.
\end{equation}
Thus, we get 
\begin{equation*}
(Q^{*})^{-l}h^n_k(0,X_0(n-l))=h^n_k(l,X_0(n-l))=\1_{l=k}.
\end{equation*}
For $l>k$, we use (\ref{a1}) and (\ref{b1}), $2^z\in ker(Q^{*})^{-1}$,
\begin{equation*}
(Q^{*})^{-l}h^n_k(0,X_0(n-l))=(Q^{*})^{-(l-k-1)}(Q^{*})^{-1}h^n_k(k,X_0(n-l))=0.
\end{equation*}
This completes the proof.
\section{Proof of Lemma \ref{nobo}}
\label{appendixD}
Now we give the proof of Lemma \ref{nobo}.
This is the answer to Exercise 5.17 in \cite{Quastel}.
By (\ref{a1})-(\ref{c1}), it is enough to check $(Q^{*})^{-1}\mathbb{P}_{RW^{*}_{l-1}=z}(\tau^{l,n}=k)=\mathbb{P}_{RW^{*}_l=z}(\tau^{l+1,n}=k)$.
Now, we assume that $X_0(n-k)=x$\footnote{Since we start from arbitrary fixed right finite initial configuration, we can write like this.} for convenience.
Then, by(\ref{in})
\begin{equation}
\begin{split}
\label{hgp}
\displaystyle(Q^{*})^{-1}\mathbb{P}_{RW^{*}_{l-1}=z}(\tau^{l,n}=k)&=2~\mathbb{P}_{RW^{*}_{l-1}=z-1}(\tau^{l,n}=k)-\mathbb{P}_{RW^{*}_{l-1}=z}(\tau^{l,n}=k)\\
&=2\left(\mathbb{P}_{RW^{*}_{l-1}=z-1}(\tau^{l,n}=k)-\frac{1}{2}~\mathbb{P}_{RW^{*}_{l-1}=z}(\tau^{l,n}=k)\right).
\end{split}
\end{equation}
By the memoryless property of geometric distribution, for $\forall y>X_0(n-k)$,
\begin{equation}
\label{hmgp}
\mathbb{P}_{RW^{*}_{l-1}=z-1}(\tau^{l,n}=k)=2^{y-X_0(n-k)} \mathbb{P}_{RW^{*}_{l-1}=z-1}(\tau^{l,n}=k,RW^{*}_k=y).
\end{equation}
Also, by (\ref{tau}),
\begin{equation}
\begin{split}
\label{kumiawase}
&\displaystyle\mathbb{P}_{RW^{*}_{l-1}=z-1}(\tau^{l,n}=k,RW^{*}_k=y)\\
&=\sum_{z-1< y_l< \dots<y_{k-1}< x}\left(\frac{1}{2}\right)^{y_l-(z-1)}\left(\frac{1}{2}\right)^{y_{l+1}-y_l}\dots\biggl(\frac{1}{2}\biggr)^{y-y_{k-1}}\times \1_{y_l\leq X_0(n-l)}\times \dots \times \1_{y_{k-1}\leq X_0(n-k+1)}\\
&=\left(\frac{1}{2}\right)^{y-(z-1)}\sum_{z-1< y_l< \dots<y_{k-1}< x} \1_{y_l\leq X_0(n-l)}\times \dots \times \1_{y_{k-1}\leq X_0(n-k+1)}.
\end{split}
\end{equation}
Note that $\displaystyle\sum_{z-1< y_l< \dots<y_{k-1}< x}=\sum_{y_l=z}^{x-1}\sum_{y_{l+1}=y_l+1}^{x-1}\dots\sum_{y_{k-1}=y_{k-2}+1}^{x-1}$, by (\ref{hmgp}) and (\ref{kumiawase}),
\begin{equation}
\begin{split}
\label{kaka}
(\ref{hgp})&=\left(\frac{1}{2}\right)^{x-z}\biggl\{\sum_{y_l=z}^{x-1}\sum_{y_{l+1}=y_l+1}^{x-1}\dots\sum_{y_{k-1}=y_{k-2}+1}^{x-1}\1_{y_l\leq X_0(n-l)}\times \dots \times \1_{y_{k-1}\leq X_0(n-k+1)}\\
&-\sum_{y_l=z+1}^{x-1}\sum_{y_{l+1}=y_l+1}^{x-1}\dots\sum_{y_{k-1}=y_{k-2}+1}^{x-1}\1_{y_l\leq X_0(n-l)}\times \dots \times \1_{y_{k-1}\leq X_0(n-k+1)}\biggr\}\\
&=\left(\frac{1}{2}\right)^{x-z}\sum_{y_{l+1}=z+1}^{x-1}\dots\sum_{y_{k-1}=y_{k-2}+1}^{x-1}\1_{z\leq X_0(n-l)}\times\1_{y_{l+1}\leq X_0(n-l-1)}\times \dots \times \1_{y_{k-1}\leq X_0(n-k+1)}.
\end{split}
\end{equation}
Since $z\leq X_0(n-l)$ was assumed,
\begin{equation}
\begin{split}
(\ref{kaka})&=\displaystyle\left(\frac{1}{2}\right)^{x-z}\sum_{z< y_{l+1}< \dots<y_{k-1}< x}\1_{y_{l+1}\leq X_0(n-l-1)}\times \dots \times \1_{y_{k-1}\leq X_0(n-k+1)}\\
&=\mathbb{P}_{RW^{*}_{l}=z}(\tau^{l+1,n}=k).
\end{split}
\end{equation}
This completes the proof.
\section{Proof of Lemma \ref{le}}
\label{appendixE}
Now we give an outline of the proof.
For $z_2\leq X_0(n)$, (\ref{grw}) can be written as
\begin{equation*}
\label{g0n}
\begin{split}
G_{0,n}(z_1,z_2)&=\mathbb{P}_{RW^{*}_{-1}=z_2}(\tau^{0,n}\leq n-1, RW^{*}_{n-1}=z_1)=\mathbb{P}_{RW_0=z_1}(\tau\leq n-1, RW_{n}=z_2)\\
&=\sum_{k=0}^{n-1}\sum_{z>X_0(k+1)}\mathbb{P}_{RW_{0}=z_1}(\tau=k, RW_k=z)Q^{n-k}(z, z_2)\\
&=\mathbb{E}_{RW_0=z_1}\left[Q^{(n-\tau)}(RW_{\tau}, z_2) \1_{\tau<n}\right],
\end{split}
\end{equation*}
where in the second equality we used the fact $Q^n(x,y)$~\eqref{qex}
represents the $n$-step transition probability of RW$_m$.
Let
\begin{equation}
\label{bg0n}
\bar{G}_{0,n}(z_1,z_2)=\mathbb{E}_{RW_0=z_1}\left[\bar{Q}^{(n-\tau)}(RW_{\tau}, z_2) \1_{\tau<n}\right].
\end{equation}
From the relation $z_2<\text{RW}_\tau<z_1$ for $\forall \t<n$ and~\eqref{eqqqb}
we see that for $\bar{G}_{0,n}(z_1,z_2)=G_{0,n}(z_1,z_2)$ for $z_2\le X_0(n)$.
Furthermore we find that $2^{-z_2}G_{0,n}(z_1,z_2)$ is polynomial in $z_2$ with degree at most $k$, and similarly $2^{-z_2}\bar{G}_{0,n}(z_1,z_2)$ is polynomial in $z_2$ with degree at most $k$ since $2^{-y_2}h^n_k(0,y_2)$ is polynomial in $y_2$ with degree at most $k$.
From the above, we find that the equality $\bar{G}_{0,n}(z_1,z_2)=G_{0,n}(z_1,z_2)$
holds for the all $z_2\in\Z$.
\section{Proof of Proposition \ref{proppp}}
\label{appendixF}
By (\ref{psi}), the lhs of~\eqref{S} becomes
\begin{equation*}
\begin{split}
\displaystyle 
&~A^{-1}_{\alpha, \beta, \gamma}(t)(\Psi^n_n)^{*}(z_1)\mid_{X_0(0)=z_2}
=A^{-1}_{\alpha, \beta, \gamma}(t)(\Psi^n_n)(z_2)\mid_{X_0(0)=z_1}
\\
&=\frac{1}{2\pi i}\oint_{\Gamma_0}dw\frac{(1-w)^n}{2^{z_2-z_1}w^{n+1+z_2-z_1}} \mathfrak{F}_{{\a,\b},\g}(w, t)
=S_{-t,-n}(z_1,z_2).
\end{split}
\end{equation*}

By (\ref{R1}) and (\ref{vq}), the lhs of~\eqref{s} is written as
\begin{align*}
&~~A_{\alpha, \beta, \gamma}(t)\sum_{z\in\mathbb{Z}}\bar{Q}^{(n)}(z_1, z)R^{-1}_{\alpha, \beta, \gamma, t}(z, z_2)\\
&=A_{\alpha, \beta, \gamma}(t)\frac{1}{2^{z_1-z_2}}\sum_{z\in\mathbb{Z}}\frac{1}{2\pi i}\oint_{\Gamma_0}dw \frac{(1+w)^{z_1-z-1}}{w^n}\cdot\frac{1}{2\pi i}\oint_{\Gamma_0}d\bar{w}\frac{\displaystyle f^{-1}_{\alpha, \beta, \gamma}(\bar{w}, t)}{\bar{w}^{z-z_2+1}}\\
&=A_{\alpha, \beta, \gamma}(t)\frac{1}{2^{z_1-z_2}}\sum_{z\in\mathbb{Z}}\frac{1}{2\pi i}\oint_{\Gamma_0}dw \frac{(1-w)^{z_1-z-1}}{w^n}(-1)^{n-1}\frac{1}{2\pi i}\oint_{\Gamma_0}d\bar{w}\frac{\displaystyle f^{-1}_{\alpha, \beta, \gamma}(\bar{w}, t)}{\bar{w}^{z-z_2+1}}\\
&=A_{\alpha, \beta, \gamma}(t)\frac{1}{2^{z_1-z_2}}\frac{(-1)^{n-1}}{2\pi i}\oint_{\Gamma_0}dw \frac{(1-w)^{z_1-z_2-1}}{w^n} f^{-1}_{\alpha, \beta, \gamma}\left(\frac{1}{1-w}, t\right).
\end{align*}
By changing variables $\displaystyle w\mapsto\frac{-w}{1-w}$, we have
\begin{multline*}
\frac{A_{\alpha, \beta, \gamma}(t)}{2^{z_1-z_2}}\frac{1}{2\pi i}\oint_{\Gamma_0}dw \frac{(1-w)^{z_2-z_1+n-1}}{w^n} f^{-1}_{\alpha, \beta, \gamma}\left(1-w, t\right)\\
\\
=\frac{1}{2\pi i}\oint_{\Gamma_0}dw\frac{(1-w)^{z_2-z_1+n-1}}{2^{z_1-z_2}w^{n}} \bar{\mathfrak{F}}_{{\a,\b},\g}(w, t)
=
\bar{S}_{-t,n}(z_1,z_2)
\end{multline*}
This completes the proof.
\section*{Acknowledgements}
We would like to thank Professor Takashi Imamura for helpful discussions 
and comments for the draft.
We also thank the referees for their valuable comments which helped to
improve the manuscript.

\end{document}